\documentclass[journal,doublecolumn]{IEEEtran}
\usepackage[T1]{fontenc}

\usepackage{threeparttable}  

\usepackage{siunitx}
\usepackage{amsfonts}
\usepackage{amsmath}
\usepackage{booktabs}
\usepackage{cancel}
\usepackage{mathrsfs}  
\usepackage{hyperref}
\usepackage{float}
\usepackage{comment}
\usepackage{subfig}
\usepackage{xcolor}
\usepackage{bbm}
\usepackage{threeparttable}
\usepackage{amssymb}
\usepackage{stfloats}
\usepackage{graphicx}
\usepackage{placeins}
\usepackage{amsthm}
\usepackage{array}
\usepackage[ruled,vlined]{algorithm2e}
\usepackage{bm}
\usepackage[square, comma, sort&compress, numbers]{natbib}
\usepackage{fancynum}
\definecolor{mygray}{gray}{.9}
\usepackage{colortbl}
\newtheorem{theorem}{Theorem}
\newtheorem{lemma}{Lemma}
\newtheorem{prop}{Proposition}
\newtheorem{defn}{Definition}
\newtheorem{remark}{Remark}

\definecolor{applegreen}{rgb}{0.55, 0.71, 0.0}

\newcommand{\Mod}[1]{\ (\mathrm{mod}\ #1)}
\newcommand{\KL}{\mathrm{KL}}
\DeclareMathOperator*{\argmin}{arg\,min}

\newcommand{\blue}[1]{{\color{black}#1}}

\begin{document}
%
\title{Relaxation-Free Min-$k$-Partition\\for PCI Assignment in 5G Networks}
\author{
	Yeqing Qiu, Chengpiao Huang, Ye Xue, Zhipeng Jiang, Qingjiang Shi, Dong Zhang, and Zhi-Quan Luo
        \thanks{The work was supported in part by the National Key Research and Development Program of China under Grant 2022YFA1003900, the National Natural Science Foundation of China under Grant 62301334, and Guangdong Major Project of Basic and Applied Basic Research under Grant 2023B0303000001.\\
        {\em (Corresponding author: Ye Xue, Qingjiang Shi)}
        \\
        Yeqing Qiu is with Shenzhen Research Institute of Big Data, and the School of Science and Engineering, The Chinese University of Hong Kong, Shenzhen, China (email: yeqingqiu@link.cuhk.edu.cn). 

        Chengpiao Huang is with the Department of Industrial Engineering and Operations Research, Columbia University, New York, USA (email: chengpiao.huang@columbia.edu).
        
        Ye Xue is with Shenzhen Research Institute of Big Data, and School of Data Science, The Chinese University of Hong Kong, Shenzhen, China (email: xueye@cuhk.edu.cn).

        Zhipeng Jiang is with the School of Mathematical Sciences, University of Chinese Academy of Sciences, Beijing, China (email: jiangzhipeng@ucas.ac.cn).
        
        Qingjiang Shi is with the School of Software Engineering, Tongji University, Shanghai, China, and Shenzhen Research Institute of Big Data, Chinese University of Hong Kong, Shenzhen, China (email: shiqj@tongji.edu.cn).

        Dong Zhang is with Huawei Technologies, Shenzhen, China (email: zhangdong48@huawei.com)

        Zhi-Quan Luo is with The Chinese University of Hong Kong, Shenzhen, China, and Shenzhen Research Institute of Big Data, Shenzhen, China (email: luozq@cuhk.edu.cn).
        }
	}%

\markboth{Journal of \LaTeX\ Class Files,~Vol.~14, No.~8, August~2015}%
{Shell \MakeLowercase{\textit{et al.}}: Bare Demo of IEEEtran.cls for IEEE Journals}



\maketitle


{

\begin{abstract}

Physical Cell Identity (PCI) is a critical parameter in 5G networks. Efficient and accurate PCI assignment is essential for mitigating mod-3 interference, mod-30 interference, collisions, and confusions among cells, which directly affect network reliability and user experience. In this paper, we propose a novel framework for PCI assignment by decomposing the problem into Min-3-Partition, Min-10-Partition, and a graph coloring problem, leveraging the Chinese Remainder Theorem (CRT). Furthermore, we develop a relaxation-free approach to the general Min-$k$-Partition problem by reformulating it as a quadratic program with a norm-equality constraint and solving it using a penalized mirror descent (PMD) algorithm. The proposed method demonstrates superior computational efficiency and scalability, significantly reducing interference while eliminating collisions and confusions in large-scale 5G networks. Numerical evaluations on real-world datasets show that our approach reduces computational time by up to 20 times compared to state-of-the-art methods, making it highly practical for real-time PCI optimization in large-scale networks. These results highlight the potential of our method to improve network performance and reduce deployment costs in modern 5G systems.

\end{abstract}

}


%
\IEEEpeerreviewmaketitle

\section{Introduction\label{intro}}

The rollout of the fifth-generation mobile networks has shown considerable promise to revolutionize the telecommunications industry, offering unprecedented speed, minimal latency, and extensive connectivity~\cite{heath2014signal, shafi20175g, gupta2015survey}. Despite the initial excitement surrounding the deployment of 5G in 2019~\cite{liu20205g}, certain regions have witnessed a decline in 5G speeds and user satisfaction \cite{pierucci2015quality}. This decline is attributed to the fact that the real-world wireless network performance depends not only on hardware capabilities, such as base stations and mobile devices but also on the configuration of network parameters \cite{luo2023srcon, liu2024}. Therefore, to improve 5G network performance, it is crucial to optimize key system parameters to match the local radio propagation environment and user distribution.

In this paper, we study the configuration of an important 5G network parameter known as the \emph{Physical Cell Identity} (PCI). It serves as a distinctive identifier for cells and enables user equipment (UE) to search for cells to connect to, identify neighbor cells for cell re-selection, and facilitate cell-to-cell handovers. A good PCI configuration can effectively reduce co-frequency interference among cells \citep{gui2018pci2}, and lead to better communication quality and user experience.

Specifically, the \emph{PCI assignment} problem involves allocating a PCI value to each cell in a 5G network. According to the 5G standard, there are $1008$ PCI values in total, which are integers from $0$ to $1007$ \citep{seetharaman2020programmable}.
Since there are typically thousands of cells in a 5G network, multiple cells may share the same PCI value.
There are multiple objectives and constraints in the 5G PCI assignment problem: minimizing mod-$3$ interference and mod-$30$ interference while eliminating collisions and confusions. 
Roughly speaking, interference, collisions, and confusions arise when certain pairs of cells share the same PCIs, or their PCIs have the same mod value.
Mathematically, PCI assignment is an NP-hard constrained combinatorial optimization problem~\cite{yang2014interference}, which becomes increasingly challenging to solve as the density of base stations increases in the continuing build-out of 5G mobile communication networks. Specifically, the solution space size is in the order of $1008^N$, where $N$ is the number of cells in the 5G network and is usually of the scale $10^5\sim 10^6$ in practice.
The prohibitively large problem size calls for the design of efficient PCI assignment algorithms.

\subsection{Related Work \label{relatedwork}}

Existing methods for PCI assignment \blue{can be mainly divided into satisfiability methods and optimization methods.}

\subsubsection{\blue{Satisfiability Methods}}

\blue{Satisfiability methods mainly focus on eliminating collisions and confusions between cells.}
\blue{To this end, they} model the wireless network as a graph, \blue{where} each \blue{vertex is} a cell and each \blue{edge} representing a collision or a confusion, and \blue{then perform graph coloring.} Bandh et al.~\cite{bandh2009graph} \blue{developed} a greedy graph coloring (GGC) approach. 
\blue{Based on the techniques from GGC,} Pratap et al.~\cite{pratap2016randomized} proposed a randomized graph coloring approach \blue{with improved computational efficiency}. However, when it comes to minimizing mod-$3$ and mod-$30$ interference, graph coloring methods are not directly applicable due to the presence of mod-$3$ and mod-$30$ operations.

\blue{
\subsubsection{Optimization Methods}

To further address mod-$3$ and mod-$30$ interference,} some heuristic algorithms have been proposed, primarily based on genetic algorithms. 
Li et al.~\cite{panxing2016pci} and Shen et al.~\cite{shen2017novel} both proposed genetic algorithm-based methods for PCI assignment, incorporating mod-$3$ interference into their respective fitness \blue{objectives}. 
\blue{
Yang et al.~\cite{yang2014interference} formulated PCI assignment as an interference self-coordination problem and proposed adaptive heuristics to minimize mod-$3$ interference.
Andrade et al.~\cite{andrade2022physical} modeled modular interference with two cell neighboring relations and introduced a memetic algorithm that combines multiparent biased random-key genetic search with local refinement.
}
However, these heuristic approaches do not have theoretical guarantees for either solution quality or computational complexity. 

\blue{Beyond heuristics, several works formulate PCI assignment as a mathematical optimization problem and design algorithms to solve it. }
Since PCI assignment is inherently a difficult NP-hard combinatorial optimization problem, most algorithms encounter significant challenges in finding high-quality solutions.
Consequently, simpler algorithms relying on greedy strategies and search methods are often employed to solve this problem, leading to suboptimal performance. 
Gui et al.~\cite{gui2018pci2} proposed a greedy algorithm to tackle the NP-hard binary quadratic programming (BQP) model for PCI assignment.
Fairbrother et al. \cite{fairbrother2018two} modeled the PCI assignment problem as a two-level graph partitioning and presented a mixed-integer programming (MIP) model to solve it. 



\blue{All existing optimization methods either incorporate weighted objectives to prioritize minimizing mod-$3$ interference, mod-$30$ interference, collisions, and confusions, or target a single arbitrary mod-$k$ interference component.
As a result, they fail to account for the structural dependence between multiple modular interferences, which often leads to a trade-off between solution quality and computational efficiency. This limitation becomes especially problematic when scaling to large networks.}

\subsection{Contributions}

In this paper, we propose a novel approach for PCI assignment that jointly addresses mod-$3$ interference, mod-$30$ interference, collisions, and confusions. Our key contributions are summarized as follows.

\begin{itemize}
\item {\bf Decomposition of PCI assignment into multi-stage graph partitioning.} Based on the Chinese Remainder Theorem (CRT), we break down PCI assignment into sequentially solving the Min-$3$-Partition, Min-$10$-Partition, and graph coloring problems. These sub-problems decouple the objectives and constraints in PCI assignment, significantly reducing the size of the solution space and enabling the scalability of the proposed method to large-scale networks. 
\item {\bf A relaxation-free reformulation of Min-$k$-Partition.} To tackle the general NP-hard Min-$k$-Partition problem, we propose an exact reformulation as a quadratic program with a norm-equality constraint over the probability simplex. To solve this reformulation, we develop a penalized mirror descent (PMD) algorithm, which is both efficient and theoretically grounded. We provide rigorous convergence guarantees \blue{that characterize the convergence rate in terms of} key system parameters. 
\item {\bf Numerical validation on real-world data.}
We conduct extensive numerical experiments on real-world datasets from dense urban 5G deployments to evaluate the performance of the proposed method. The results demonstrate that our approach outperforms state-of-the-art baselines in terms of interference minimization, collision and confusion elimination, and computational efficiency. These findings highlight the practical applicability of our method in optimizing PCI assignments for large-scale, real-world 5G networks.
\end{itemize}

\emph{Synopsis.} The subsequent sections of the paper are structured as follows. Section \ref{sec:prob} introduces the system model and formulates the PCI assignment problem. Section \ref{sec:decomposition} describes our approach for the PCI assignment based on the CRT, decomposing the problem into sequentially solving Min-$k$-Partition problems and a graph coloring problem. Building upon this decomposition, Section \ref{sec:prop} presents a relaxation-free reformulation of Min-$k$-Partition and details the proposed PMD algorithm along with convergence analysis. 
{Section \ref{sec:ext} extends the proposed framework to a more practical case where the PCI of a subset of cells is optimized.}
Section \ref{sec:exp} presents numerical experiments on synthetic and real-world data. Finally, Section \ref{sec:con} concludes the paper.

\emph{Notations.} 
We use $\mathbb{Z}$ and $\mathbb{R}$ to denote the set of integers and real numbers, respectively. 
Moreover, $\mathbb{Z}_+$ denotes the set of non-negative integers, $\mathbb{Z}_{k}$ is the set integers ranging from $0$ to $k-1$, $\mathbb{R}^N$ is the set of $N$-dimensional real-valued vectors, and $\mathbb{Z}_k^N$ is the set of $N$-dimensional integer vectors with entries from $\mathbb{Z}_k$. For $a,b\in\mathbb{Z}$ and $k\in\mathbb{Z}_+$, we write $a\equiv b\pmod{k}$ if $a-b=nk$ for some $n\in\mathbb{Z}$. The mod $k$ value of $a$, denoted by $m_k(a)$, is the unique integer $r\in\mathbb{Z}_k$ such that $a\equiv r\pmod{k}$. For $x\in\mathbb{R}$, we use $\lfloor x\rfloor$ to denote the largest integer less than or equal to $x$. We use lowercase and uppercase boldface letters to denote column vectors and matrices, respectively. The notation $[\bm{A}]_{i,j}$ represents the element in the $i$-th row and $j$-th column of a matrix $\bm{A}$. 
The transpose of a matrix $\bm{A}$ is denoted by $\bm{A}^T$. The trace of a matrix $\bm{A}$ is denoted by $\text{Tr}(\bm{A})$. The $\ell_p$-norm of a vector $\bm{v}$ is denoted by $\|\bm{v}\|_p$. The $\ell_{p,q}$-norm of a matrix $\bm{A}$ is given by $\Vert\bm{A}\Vert_{p,q}=\bigg(\sum_{i=1}^n\big(\sum_{j=1}^n\vert[\bm{A}]_{i,j}\vert^p\big)^{\frac{q}{p}}\bigg)^{\frac{1}{q}}$. The notation $\langle \cdot, \cdot\rangle$ represents the generalized inner product. The vector $\bm{1}_N$ denotes the $N$-dimensional all-one vector. The function $\mathbbm{1}\{A\}$ serves as the indicator function for event $A$. The cardinality of a set $\mathcal{S}$ is denoted by $\vert\mathcal{S}\vert$.

\begin{figure}
    \centering
    \includegraphics[width=\linewidth]{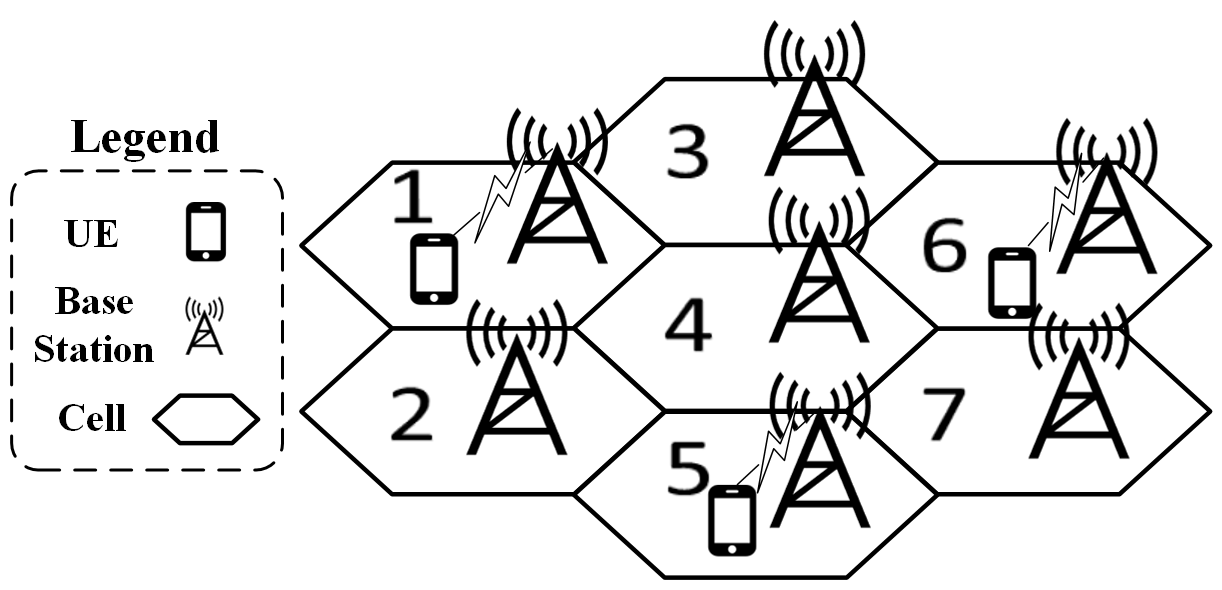}
    \caption{Illustration of a 5G wireless network comprising $N=7$ co-frequency cells.
    Each hexagon is a cell, and two cells are neighboring cells if they share an edge. For example, $(1,4)\in\mathcal{E}_1$ indicates a collision between cell $1$ and cell $4$ when they share the same PCI. 
    Moreover, since $(4,1), (4,7)\in\mathcal{E}_1$, this implies $(1,7)\in\mathcal{E}_2$, and confusion arises when cells $1$ and $7$ are assigned the same PCI. 
     When the PCIs of cells $1$ and $2$ share the same mod-$3$ value, mod-$3$ interference occurs with a magnitude of $[\bm{W}]_{1,2}$. Mod-$30$ interference follows a similar pattern.
     }
    \label{fig:sysmod}
\end{figure}

\section{System Model and Problem Formulation \label{sec:prob}}

In this section, we formally define the PCI assignment problem.

\subsection{System Model}

We model a 5G wireless network with $N$ cells as a graph, where the vertices are given by $\mathcal{V}=\{1, 2, \cdots, N\}$, and vertex $i$ represents the $i$-th cell. Between cells, there is interference and a neighboring relation. The interference between cells is encoded in an \emph{interference matrix} $\bm{W}\in\mathbb{R}^{N\times N}$, where $[\bm{W}]_{i,j}=[\bm{W}]_{j,i} \ge 0$ is the interference between cell $i$ and cell $j$. For two cells $i$ and $j$, we say that $j$ is a \emph{neighboring cell} of $i$ if the signal of cell $j$ can be received by devices connected to cell $i$. The set of neighboring cell pairs is denoted by $\mathcal{E}_1\subseteq\mathcal{V}\times\mathcal{V}$. We further define the \emph{second-order neighboring set}
\[
\mathcal{E}_2 = \{ (i,j) \in \mathcal{V}\times\mathcal{V} \mid \text{$\exists\ell\in\mathcal{V}$ such that $(\ell,i),(\ell,j)\in\mathcal{E}_1$} \},
\]
which consists of cell pairs that share the same neighbor. We let $\mathcal{E}=\mathcal{E}_1\cup\mathcal{E}_2$. Similar to \cite{andrade2022physical}, we assume that $\bm{W}$, $\mathcal{E}_1$, and $\mathcal{E}_2$  have been estimated and are known.

\subsection{Problem Formulation}

The task of the 5G PCI assignment problem is to assign to each cell $i\in\mathcal{V}$ an integer $\text{PCI}_i\in\mathbb{Z}_{1008}$, called the \emph{PCI value}, with the goal of minimizing mod-$3$ and mod-$30$ interference and eliminating collisions and confusions. 
\begin{itemize}

    \item \blue{\emph{Mod-$3$ interference} occurs between two cells $i$ and $j$, if their PCIs have the same mod-$3$ value, i.e.~$\text{PCI}_i \equiv \text{PCI}_j \pmod{3}$. In this case, the magnitude of the mod-$3$ interference is given by $[\bm{W}]_{i,j}$. 
    Physically, the mod-$3$ value is used for the scheduling of the physical downlink and uplink shared channel as well as the sounding reference signal. An overlap of mod-$3$ value leads to conflicts in resource scheduling \cite{Huawei5GPlanning}.
    \item \emph{Mod-$30$ interference} occurs when two cells have the same mod-$30$ value, i.e.~$\text{PCI}_i \equiv \text{PCI}_j \pmod{30}$. In this case, the magnitude of the mod-$30$ interference is given by $[\bm{W}]_{i,j}$. Physically, an overlap of mod-$30$ value indicates the reuse of the same Zadoff-Chu root sequences for uplink demodulation reference signals and sounding reference signals, leading to severe sequence collision~\cite{Huawei5GPlanning}.
    }
    \item A \emph{collision} arises when two neighboring cells $i$ and $j$ share the same PCI, that is, $(i,j)\in\mathcal{E}_1$ and $\text{PCI}_i=\text{PCI}_j$. It prevents a UE from determining the serving cell and establishing a connection \cite{andrade2022physical}.
    \item A \emph{confusion} arises when two cells $(i,j)\in\mathcal{E}_2$ share the same PCI. In other words, this means that two cells $i,j$ sharing a common neighbor have the same PCI. This can potentially cause a UE confusion during cell handovers.
\end{itemize}
In Figure \ref{fig:sysmod}, we provide a visualization of a 5G network and illustrate the concepts of interference, collisions, and confusions. 
\blue{Although the physical mechanisms that generate mod-$3$ and mod-$30$  interference differ, we represent both with the same interference matrix $\bm{W}$. In practice, PCI assignment is performed at the network layer, where the only reliable source of interference information is the MR data collected from live networks. MR data already aggregates all signal-level degradations, so a unified matrix $\bm{W}$ is aligned with common industrial practice and ensures that the optimization concentrates on measurable interference patterns rather than unobservable physical-layer particulars~\cite{CN105744548A,CN103906078A}.}

Mathematically, the PCI assignment problem can be formulated as
\begin{subequations}    
\begin{align}
    \text{(P1)}\quad \min_{\textbf{PCI}}& \quad 
    \Bigg( 
    \sum_{i,j=1}^N[\bm{W}]_{i,j} \cdot \mathbbm{1}\{\text{PCI}_i\equiv\text{PCI}_j\Mod{3}\}, \label{obj-mod3} \\
    & \quad\ \  \sum_{i,j=1}^N[\bm{W}]_{i,j} \cdot \mathbbm{1}\{\text{PCI}_i\equiv\text{PCI}_j\Mod{30}\} \Bigg), \label{obj-mod30} \\
    \text{s.t.}&\quad  \text{PCI}_i\neq\text{PCI}_j , \quad \forall (i,j)\in \mathcal{E}, \label{cons-coll-conf} \\
&\quad  \text{PCI}_i \in \mathbb{Z}_{1008}, \quad \forall i\in \mathcal{V}, \label{cons-range}
\end{align}
\end{subequations}
where $\textbf{PCI} = (\text{PCI}_1,...,\text{PCI}_N)$.
This is a multi-objective optimization problem where the objectives are to minimize the mod-$3$ interference \eqref{obj-mod3} and the mod-$30$ interference \eqref{obj-mod30}, and the constraints include eliminating collisions and confusions \eqref{cons-coll-conf} and the PCI range requirement \eqref{cons-range}.

\section{A Multi-Stage Decomposition Approach \label{sec:decomposition}}

In this section, we present our approach to PCI assignment, which decomposes the problem into multiple stages of graph partitioning and graph coloring, effectively decoupling the multiple objectives and constraints.

\subsection{Decomposition Framework}
For each cell $i$, we can decompose its PCI value as
\[
\text{PCI}_i = 30q_i + r_i,
\]
where $q_i=\lfloor\text{PCI}_i/30\rfloor\in\mathbb{Z}_+$ is the quotient, and $r_i=m_{30}(\text{PCI}_i)\in\mathbb{Z}_{30}$ is the mod-$30$ value of $\text{PCI}_i$. Thus, a PCI assignment $\textbf{PCI}$ is equivalently characterized by the mod-$30$ value assignment $\bm{r}=(r_1,...,r_N)$ and the quotient assignment $\bm{q}=(q_1,...,q_N)$. The PCI assignment problem (P1) can then be reformulated as
\begin{subequations}
\begin{align}
    \text{(P2)}\quad \min_{\bm{q},\bm{r}} & \quad \Bigg(  \sum_{i,j=1}^N [\bm{W}]_{i,j} \cdot \mathbbm{1}\{r_i\equiv r_j\Mod{3}\}, \label{objs:1}\\
    & \quad\ \  \sum_{i,j=1}^N [\bm{W}]_{i,j} \cdot \mathbbm{1}\{r_i = r_j\} \Bigg), \label{objs:2}\\
    \text{s.t.}&\quad q_i \neq q_j \text{ or } r_i \neq r_j, \quad \forall (i,j)\in \mathcal{E}, \notag \\
    &\quad \bm{q}\in \mathbb{Z}_+^N, \notag \\
    & \quad \bm{r}\in \mathbb{Z}_{30}^N, \notag \\
    & \quad 30\times \bm{q} + \bm{r} \in \mathbb{Z}_{1008}^N. \notag
\end{align}
\end{subequations}
Note that the objective \eqref{objs:2} does not involve modulo operations because $r_i\in\mathbb{Z}_{30}$. Moreover, since both objectives \eqref{objs:1} and \eqref{objs:2} depend solely on $\bm{r}$, then we can break down the problem into first optimizing over $\bm{r}$ to minimize both mod-$3$ interference and mod-$30$ interference, and then assigning $\bm{q}$ to eliminate collisions and confusions. This leads to the following two sub-problems:
\begin{itemize}
    \item {\bf Sub-problem 1: Mod-$30$ value assignment $\bm{r}$.} We determine the mod-$30$ value assignment $\bm{r}$ by minimizing the mod-$3$ and mod-$30$ interference:
\begin{align}
     \min_{\bm{r}}&\quad \Bigg( \sum_{i,j=1}^N [\bm{W}]_{i,j} \cdot \mathbbm{1}\{r_i\equiv r_j\Mod{3}\},  \notag \\
     &\quad\ \ \sum_{i,j=1}^N [\bm{W}]_{i,j} \cdot \mathbbm{1}\{r_i = r_j\} \Bigg), \notag \\
    \text{s.t.}&\quad  \bm{r}\in \mathbb{Z}_{30}^N.\notag
\end{align}
    \item {\bf Sub-problem 2: Quotient assignment $\bm{q}$.} Suppose we have found a mod-$30$ value assignment \blue{$\bm{r}$} for sub-problem 1. Then, we need to find a quotient assignment $\bm{q}$ that eliminates collisions and confusions. This is a feasibility problem:
    \begin{align*}
    \text{find}&\quad \bm{q}=(q_1, \cdots, q_N), \label{subprob:coloring}\\
    \text{s.t.}&\quad q_i\neq q_j, \quad \forall (i, j)\in\mathcal{E}\text{ with } \widetilde{r}_i = \widetilde{r}_j , \notag\\
    & \quad \bm{q}\in \mathbb{Z}_+^N, \notag\\ 
    & \quad 30 \bm{q}+\blue{\bm{r}}\in \mathbb{Z}_{1008}^N. \notag
    \end{align*} 
\end{itemize}

In the following two subsections, we will further decompose sub-problems 1 and 2 into smaller problems, which greatly reduces the problem size.

\subsection{Two-Stage Graph Partitioning for Mod-$30$ Value Assignment}\label{sec:mod-3 and mod-10}

We observe that the mod-$3$ interference minimization objective in sub-problem 1 depends solely on the mod-$3$ values of the cells' PCIs, regardless of their specific mod-$30$ values. \blue{In light of this observation, we will first determine the mod-$3$ values of the cells' PCIs to minimize the mod-$3$ interference. Based on the mod-$3$ values, we will then decide the mod-$30$ values to minimize the mod-$30$ interference. To this end, we invoke the Chinese Remainder Theorem (CRT)~\cite{CRT}, which states that two mod values of a number, say the mod-$n_1$ and mod-$n_2$ values, uniquely determine its mod-$n_1n_2$ value as long as $n_1$ and $n_2$ are coprime.}
Accordingly, the mod $30$ value $r_i = m_{30}(\text{PCI}_i)$ is uniquely determined by its mod-$3$ value $r^{(3)}_i = m_3(\text{PCI}_i)$ and mod-$10$ value $r^{(10)}_i = m_{10}(\text{PCI}_i)$. 
Therefore, we can further decompose sub-problem 1 by first finding a mod $3$ value assignment $\bm{r}^{(3)} = (r^{(3)}_1,...,r^{(3)}_N)$ to minimize the mod $3$ interference, and then finding a mod $10$ value assignment $\bm{r}^{(10)} = (r^{(10)}_1,...,r^{(10)}_N)$ to minimize the mod $30$ interference. Mathematically, this corresponds to the following two sub-problems:
\begin{itemize}
    \item {\bf Sub-problem 1.1: Mod-$3$ value assignment $\bm{r}^{(3)}$.} We determine the mod-$3$ value assignment $\bm{r}^{(3)}$ by minimizing the mod-$3$ interference:
\begin{align*}
    \text{(S1.1)}\quad\min_{\bm{r}^{(3)}} & \quad \sum_{i,j=1}^N [\bm{W}]_{i,j} \cdot \mathbbm{1} \{ r^{(3)}_i = r^{(3)}_j \},   \\
    \text{s.t.}&\quad  \bm{r}^{(3)}\in \mathbb{Z}_{3}^N.
\end{align*}
Note that the objective is indeed the mod-$3$ interference minimization objective in sub-problem 1, because $r_i\equiv r_j\Mod{3}$ if and only if $r_i^{(3)} = r_j^{(3)}$. Sub-problem 1.1 is in fact a \emph{Min-$3$-Partition} problem which is NP-hard \cite{boros1991max, kann1996hardness, ma2017multiple}. We will propose an algorithm for solving it in Section \ref{sec:prop}.

\item {\bf Sub-problem 1.2: Mod-$10$ value assignment $\bm{r}^{(10)}$.} Suppose we have found a mod-$3$ value assignment \blue{$\bm{r}^{(3)}$} in sub-problem 1.1. We now need to find a mod-$10$ value assignment $\bm{r}^{(10)}$ that minimizes the mod-$30$ interference. An important observation is that if two cells are assigned different mod-$3$ values, then regardless of the mod-$10$ value assignment, they must have different mod-$30$ values and hence no mod-$30$ interference. We can take advantage of this observation to reduce the problem size. Specifically, we partition the network into three \emph{partitions} $\mathcal{V}_0$, $\mathcal{V}_1$, $\mathcal{V}_2$ according to the mod-$3$ values of the cells:
\[
\mathcal{V}_{r} = \{i\in\mathcal{V} \mid \blue{r^{(3)}_i} = r\},\quad r\in\mathbb{Z}_3.
\]
Then we perform mod-$10$ value assignment on each partition \emph{in parallel}: for each $r\in\mathbb{Z}_3$, we determine the mod-$10$ value assignment $r^{(10)}_i$ for cells $i$ in 
$\mathcal{V}_r$ by solving
\begin{align*}
\text{(S1.2.$r$)}\quad\min&\quad\sum_{i,j\in\mathcal{V}_r}[\bm{W}]_{i,j}\cdot\mathbbm{1} \{ r^{(10)}_i = r^{(10)}_j \}, \\
\text{s.t.}&\quad r_i^{(10)} \in \mathbb{Z}_{10},\quad\forall i\in\mathcal{V}_r.
\end{align*}

As each $\mathcal{V}_r$ can be much smaller than $\mathcal{V}$, this approach of optimizing over $\mathcal{V}_r$ in parallel is more efficient than optimizing over the whole set $\mathcal{V}$ directly. The problem (S1.2.$r$) is in fact a \emph{Min-$10$-Partition} problem which is NP-hard \cite{boros1991max, kann1996hardness, ma2017multiple}. We will propose an algorithm for solving it in Section \ref{sec:prop}.
\end{itemize}
After solving sub-problems 1.1 and 1.2, we obtain a mod-$3$ value assignment $\blue{\bm{r}}^{(3)}\in\mathbb{Z}_3^N$ and a mod-$10$ value assignment $\blue{\bm{r}}^{(10)}\in\mathbb{Z}_{10}^N$. We then invoke the CRT to get a mod-$30$ value assignment $\blue{\bm{r}}\in\mathbb{Z}_{30}^N$.

\subsection{Graph Coloring for Quotient Assignment}\label{sec:quotient}

Finally, we study sub-problem 2. Recall that the goal is to determine a quotient assignment $\bm{q}$ that eliminates collisions and confusions. Similar to sub-problem 1.2, a key observation for reducing the problem size is that if two cells are assigned different mod-$30$ values, then regardless of the quotient assignment, they must have different PCI values and thus no collision or confusion. In light of this, we divide the network into $30$ \emph{clusters} $\mathcal{C}_0,\mathcal{C}_1,...,\mathcal{C}_{29}$ according to the mod-$30$ value assignment:
\[
\mathcal{C}_r = \{ i \in \mathcal{V} \mid \blue{r}_i = r \},\quad r\in\mathbb{Z}_{30}.
\]
We then perform quotient assignment on each $\mathcal{C}_r$ in parallel: for each $r\in\mathbb{Z}_{30}$, we determine the quotient assignment $q_i$ for cells $i$ in $\mathcal{C}_r$ by solving
\begin{align*}
    \text{(S2.$r$)}\quad \text{find} &\quad q_i,\quad i\in\mathcal{C}_r, \\
    \text{s.t.} &\quad q_i \neq q_j,\quad \forall (i,j)\in \mathcal{E}\cap(\mathcal{C}_r\times\mathcal{C}_r), \\
    &\quad q_i \in \mathbb{Z}_+,\quad\forall i\in\mathcal{C}_r, \\
    &\quad 30q_i + \blue{r}_i \in \mathbb{Z}_{1008}, \quad\forall i\in\mathcal{C}_r.
\end{align*}
This problem can be reformulated as a \emph{graph coloring problem} \cite{bandh2009graph}. In the graph coloring problem, one is given an undirected graph and seeks to assign a \emph{color} to each vertex so that vertices sharing an edge have different colors. It is easy to see that (S.2.$r$) can be solved by performing graph coloring over a graph with vertices $\mathcal{C}_r$ and edges $\mathcal{E}\cap(\mathcal{C}_r\times\mathcal{C}_r)$, and setting up a one-to-one correspondence between the colors and values of the quotient. To solve the graph coloring problem, we directly use GGC from \cite{bandh2009graph, coloring}. 

\blue{The graph coloring method does not always guarantee that the PCI range constraint $30q_i+\blue{r}_i \in\mathbb{Z}_{1008}$ is satisfied. When too many colors are used, some cells may be assigned overly large $q_i$, violating the PCI range constraint. To tackle this issue, we develop a heuristic post-processing procedure that re-assigns quotients $q_i$ to cells $i$ whose PCIs exceed $1007$, at the cost of potentially increasing the numbers of collisions and confusions. For each of these cells $i$, we fix its mod-$30$ value $r_i$, and update $q_i$ to a value $q\in\mathbb{Z}$ that satisfies $30q + r_i \leq 1007$ and minimizes the number of collisions and confusions with cells that already satisfy the PCI range constraint. We process the constraint-violating cells in descending order of the total number of neighbors and second-order neighbors, so as to prioritize highly constrained regions. 
The full algorithm is provided in Section A of the supplementary material~\cite{qiu2025relaxation}.}

\begin{figure}
    \centering
    \includegraphics[width=\linewidth]{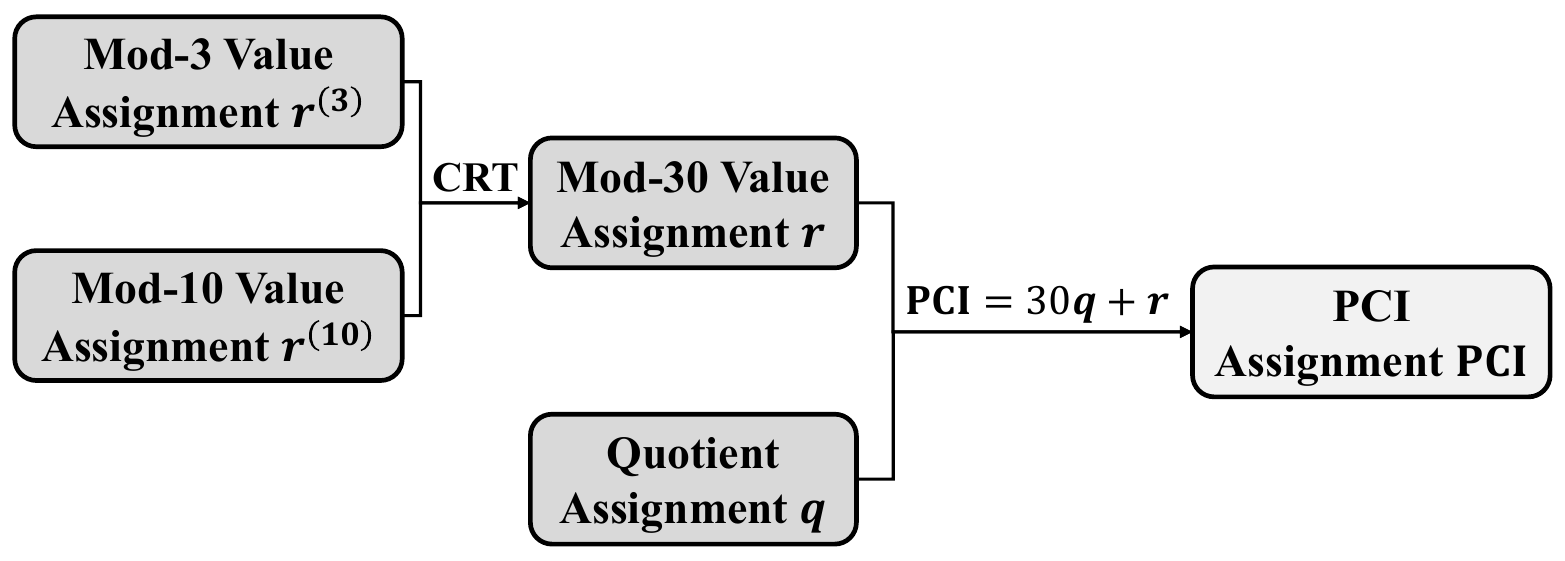}
    \caption{Illustration of the proposed PCI assignment framework. The framework comprises two primary steps: mod-30 value assignment and quotient assignment. In the mod-30 value assignment, we first solve Min-$3$-Partition to obtain mod-3 values $\blue{\bm{r}}^{(3)}$ and then solve Min-$10$-Partition to determine the mod-10 values $\blue{\bm{r}}^{(10)}$. Subsequently, we derive the mod-30 values $\blue{\bm{r}}$ based on the CRT.
    In the quotient assignment, solve the graph coloring problems based on greedy graph coloring to determine the quotient values $\blue{\bm{q}}$. 
    Finally, the PCI value $\blue{\textbf{PCI}}$ for each cell is computed based on the mod values $\blue{\bm{r}}$ and quotient values $\blue{\bm{q}}$. 
}
    \label{fig:framework}
\end{figure}

\subsection{Summary}

Below we summarize our approach for PCI assignment, which decomposes the problem into assigning mod-$3$ values $\bm{r}^{(3)}$, mod-$10$ values $\bm{r}^{(10)}$, and quotients $\bm{q}$. An illustration is given by Figure \ref{fig:framework}.

\begin{enumerate}
    \item {\bf Mod-$3$ value assignment.} Solve the Min-$3$-Partition problem (S1.1) to obtain a mod-$3$ value assignment $\blue{\bm{r}}^{(3)}\in\mathbb{Z}_3^N$. This step minimizes the mod-$3$ interference. The algorithm for solving Min-$3$-Partition is introduced in Section \ref{sec:prop}. 
    \item {\bf Mod-$10$ value assignment.} Solve the Min-$10$-Partition problems (S1.2.$r$), $r=0,1,2$ in parallel to obtain a mod-$10$ value assignment $\blue{\bm{r}}^{(10)}\in\mathbb{Z}_{10}^N$. This step minimizes the mod-$30$ interference. The algorithm for solving Min-$10$-Partition is introduced in Section \ref{sec:prop}.
    \item \textbf{Mod-$30$ value assignment.} With $\blue{\bm{r}}^{(3)}$ and $\blue{\bm{r}}^{(10)}$, apply the CRT to obtain a mod-$30$ value assignment $\blue{\bm{r}}\in\mathbb{Z}_{30}^N$.
    \item \textbf{Quotient assignment.} Solve the graph coloring problems (S2.$r$), $r=0,1,...,29$ in parallel to obtain a quotient assignment $\blue{\bm{q}}\in\mathbb{Z}_+^N$. This step eliminates collisions and confusions.
    \item \textbf{PCI assignment.} The PCI assignment is given by \[
    \blue{\textbf{PCI}}=30\blue{\bm{q}}+\blue{\bm{r}}.
    \]
\end{enumerate}

\blue{
While our main focus is the joint minimization of mod-$3$ and mod-$30$ interference, our decomposition framework naturally extends to arbitrary single mod-$k$ interference by expressing each PCI as $\text{PCI}_i = kq_i + r_i$, solving a Min-$k$-Partition for $r_i$, and applying graph coloring for $q_i$. Furthermore, it supports joint optimization over multiple moduli, provided they have the coprime structure required by the CRT. This yields a unified and scalable framework for systematically decomposing the multiple modular interferences minimization.
}

\section{Relaxation-free Min-$k$-Partition and \blue{PMD} \label{sec:prop}}

In this section, we propose a relaxation-free approach to the general Min-$k$-Parition problem, which appears in the mod-$3$ value assignment and mod-$10$ value assignment in our PCI assignment approach. The Min-$k$-Partition problem is NP-hard \cite{boros1991max, kann1996hardness, ma2017multiple}. Current approaches, such as semidefinite relaxation \cite{delorme1993laplacian, lovasz1979shannon, lovasz1991cones, poljak1995nonpolyhedral, shor1987quadratic} and rank-$2$ relaxation methods \cite{burer2002rank} are effective only for Min-$2$-Partition, which is inadequate for large-scale 5G wireless networks. Furthermore, the inherent relaxation error associated with these methods often leads to suboptimal solutions. We will now formally define the Min-$k$-Partition problem and introduce our relaxation-free approach. 

\subsection{Relaxation-Free Min-$k$-Partition \label{mod-k}}
The general \emph{Min-$k$-Partition} problem aims to divide the vertices in a graph into $k$ subsets and minimize the interference between vertices within each subset. Using our notation, it can be formulated as
\begin{align*}
     \underset{\bm{r}}{\min} & \quad \sum_{i,j=1}^N[\bm{W}]_{i,j}\cdot\mathbbm{1} \{ r_i = r_j \},  \\
    \text{s.t. }&\quad  \bm{r}\in \mathbb{Z}_{k}^N.
\end{align*}

In our PCI assignment approach, we take $k=3$ (resp.~$k=10$) for mod-$3$ (resp.~mod-$10$) value assignment, where each mod-$3$ (resp.~mod-$10$) value corresponds to a subset of cells. Here, we have slightly abused the notation $N$, as the mod-$10$ assignment is done over each mod-$3$ partition instead of the whole set of cells.

To solve Min-$k$-Partition, we use a one-hot encoding scheme for each of the $k$ subsets.
Specifically, we set 
\[
    \bm{x}_i=\bm{e}_{r_i},
\]
where $\bm{e}_{r_i}$ denotes the $k$-dimensional unit vector whose $r_i$-th entry is $1$, and all other entries are $0$. Then, we can reformulate Min-$k$-Partition as
\begin{align}
    \text{(P3)}\quad\underset{\bm{X}=[\bm{x}_1, \bm{x}_2, \cdots, \bm{x}_N]}{\min}& \quad  \sum_{i=1}^N\sum_{j=1}^N [\bm{W}]_{i,j}{\bm{x}}_i^T \bm{x}_j \label{opt:modk:2}  \\
    \text{s.t.}& \quad \bm{x}_i\in \{\bm{e}_1, \cdots, \bm{e}_k\}, \quad \forall i\in \mathcal{V}. \notag
\end{align}    
The constraint $\bm{x}_i\in \{\bm{e}_1, \cdots, \bm{e}_k\}$ in problem (P3) is discrete and nonconvex. 
\blue{As \( \bm{e}_1, \cdots, \bm{e}_k \) are precisely the $k$-sparse vertices of the probability simplex \( \Delta_k = \{(x_1,\cdots,x_k)^T \in \mathbb{R}^k \mid \sum_{i=1}^k x_i = 1,\, x_i \ge 0 \} \), we utilize the following norm-based characterization of sparsity proposed by \cite{wang2021clustering}. 
}
\begin{lemma}[Norm Condition~\cite{wang2021clustering}]
Fix $q>p\ge 1$. Any non-zero vector $\bm{x}=(x_1, \cdots, x_k)^T\in\mathbb{R}^k$ has one non-zero entry if and only if $\Vert \bm{x}\Vert_p=\Vert \bm{x}\Vert_q$. 
\label{prop:1}
\end{lemma}
\blue{
\begin{proof}
    See Appendix \ref{proof:prop:1}.
\end{proof}
}

Using Lemma \ref{prop:1}, we arrive at the following relaxation-free reformulation of Min-$k$-Partition:
\begin{subequations}
\begin{align}
    \text{(P4)}\quad\underset{\bm{X}=[\bm{x}_1,\cdots, \bm{x}_N]}{\min}& \quad  \sum_{i=1}^N\sum_{j=1}^N [\bm{W}]_{i,j}{\bm{x}}_i^T \bm{x}_j  \label{opt:modk:3}  \\
    \text{s.t. }& \quad \Vert \bm{x}_i\Vert_p=\Vert \bm{x}_i\Vert_q, \quad \forall i\in \mathcal{V}  \label{opt:modk:3:condi1}, \\
    & \quad \bm{x}_i\in \Delta_k, \quad \forall i\in\mathcal{V}, \label{opt:modk:3:condi2}
\end{align}    
\end{subequations}where $1\le p<q$.

\subsection{Penalized Mirror Descent Algorithm\label{alg}}

Given the equivalence of (P3) and (P4), we propose a PMD algorithm to solve (P4) by constructing a sequence of penalized problems and solving them by mirror descent (MD).

\subsubsection{Norm-Equality Penalization}
In view of its discrete feasible region, directly optimizing (P4) may exhibit sensitivity to the initial point.
Therefore, we apply the penalty method \cite{nocedal1999numerical, selesnick2014sparse, shi2020penalty, shi2020penalty2} to gradually arrive at a binary solution, mitigating sensitivity to initial points. Consider the following penalized version of (P4):
\begin{subequations}
\begin{align} 
    \underset{\bm{X}=[\bm{x}_1,\cdots,\bm{x}_N]}{\min}& \quad  \sum_{i=1}^N\sum_{j=1}^N[\bm{W}]_{i,j}\bm{x}_i^T\bm{x}_j + \frac{\rho}{v}\sum_{i=1}^N(\Vert \bm{x}_i\Vert_p^v-\Vert\bm{x}_i\Vert_q^v),  \label{obj:penalty}\\
    \text{s.t.}& \quad \bm{X}\in \Delta_k^N,
\end{align}      
\label{opt:modk:4.5:condi}
\end{subequations}
where $v>0$, and $\Delta_k^N$ is the Cartesian product of $N$ copies of the $k$-dimensional probability simplex. We note that $\|\bm{x}\|_p^v - \|\bm{x}\|_q^v \ge 0$ for all $\bm{x}\in\mathbb{R}^k$ whenever $q\ge p \ge 1$. 
\blue{To facilitate efficient computation, we adopt the smooth setting in \cite{wang2021clustering} with $p=1$, $q=2$, and $v=2$, so that the resulting penalty is continuously differentiable. This enables the use of scalable first-order optimization algorithms, which are significantly more efficient than proximal methods typically required for non-smooth penalties.} The penalized version of (P4) can be written as
\begin{equation}
\begin{aligned} \label{opt:modk:4:condi}
    \text{(P5)}\quad\underset{\bm{X}}{\min}& \quad  \text{Tr}\left(\bm{X}\left(\bm{W}-\frac{\rho}{2}\bm{I} \right)\bm{X}^T \right) + \frac{N\rho}{2}\triangleq F(\bm{X};\rho)  \\
    \text{s.t. }& \quad \bm{X}\in \Delta_k^N,
\end{aligned}    
\end{equation}
where $\rho>0$ is the penalty parameter. The detailed derivation of problem (P5) is provided in Appendix \ref{app:0.75}.

Next, we present the following Theorem \ref{theo:1} to establish the theoretical foundation for the exactness of the proposed penalty method.

\begin{theorem} [Exactness of the Penalized Problem] 
Let $\rho_0=\max(2\lambda_{\max}(\bm{W}), 0)$. When $\rho>\rho_0$, the problem (P5) has the same optimal solutions as (P4). 
\label{theo:1}
\end{theorem}
\begin{proof}
    See Appendix \ref{app:1}.
\end{proof}

\begin{remark} Although penalization is often considered a form of relaxation, Theorem \ref{theo:1} rigorously establishes that the proposed framework remains relaxation-free under the condition that the penalty parameter $\rho$ is sufficiently large. Specifically, when $\rho$  exceeds a certain threshold determined by the problem constants, the penalization guarantees equivalence to the original problem. This ensures that the proposed approach maintains the exactness of the original optimization formulation while benefiting from the computational advantages of penalization.
\end{remark}

\subsubsection{Solving the Penalized Problem based on MD}
For a fixed $\rho$, (P5) comprises optimizing a non-convex objective function over the product of probability simplices. The MD algorithm is well-suited for this task. 
In particular, the $m$-th iteration of MD is 
\begin{align}
    \bm{X}^{(m+1)}=\argmin_{\bm{X}\in\Delta_k^N} \big\{ & F(\bm{X}^{(m)};\rho) \notag\\ & +\left\langle\nabla F(\bm{X}^{(m)};\rho),\bm{X}
    -\bm{X}^{(m)}\right\rangle \notag\\
    &+\frac{1}{t_m}B_\omega(\bm{X},  \bm{X}^{(m)}) \big\}, \notag
    \label{for:prox:md} 
\end{align} 
where $t_m$ is the step-size of the $m$-th iteration, $B_\omega$ is the Bregman divergence with distance generating function $\omega$: 
\begin{align}
    B_\omega(\bm{X}, \bm{Y}) = \omega(\bm{X})-\omega(\bm{Y})-\left\langle\nabla\omega(\bm{Y}),\bm{X}-\bm{Y}\right\rangle.
\end{align}
To account for the geometric structure of the probability simplex, we employ the KL divergence as the specific Bregman divergence in our approach. The KL divergence effectively quantifies the disparity between probability vectors while capturing the variations within a distribution. It is defined as
\begin{align}
    \KL(\bm{X},\bm{X}^{(m)})&=\sum_{j=1}^N\sum_{i=1}^k[\bm{X}]_{i,j}\log\frac{[\bm{X}]_{i,j}}{[\bm{X}^{(m)}]_{i,j}} . \label{divergence}
\end{align}
The following proposition indicates that each iteration of MD for solving P5 has a closed-form solution.
\begin{prop}[Closed-form Iteration of MD \cite{beck2017first}]\label{prop:optimal}
   The $m$-th iteration of MD for solving (P5) is given by
       \begin{multline}
               \bm{X}^{(m+1)}
   =\argmin_{\bm{X}\in\Delta_k^N} \bigg\{ \left\langle \bm{X}^{(m)}(2\bm{W}-\rho\bm{I}),\bm{X}\right\rangle  \\
    +\frac{1}{t_m} \KL(\bm{X},\bm{X}^{(m)}) \bigg\}.
    \label{alg:MD}
       \end{multline}
    Equivalently, the update can be computed entrywise via
    \begin{multline}
    [\bm{X}^{(m+1)}]_{i,j}\\
    =\frac{[\bm{X}^{(m)}]_{i,j}\exp(-t_m[\bm{X}^{(m)}(2\bm{W}-\rho\bm{I})]_{i,j})}{\sum_{p=1}^k [\bm{X}^{(m)}]_{p,j}\exp(-t_m[\bm{X}^{(m)}(2\bm{W}-\rho\bm{I})]_{p,j})}.
    \label{KLMDup}
    \end{multline}
\end{prop}

\subsubsection{Convergence Analysis}

MD has been extensively studied in stochastic and non-smooth optimization problems~\cite{zhang2018convergence, zhou2017stochastic, zhou2017mirror}. While these studies provide valuable insights into the convergence behavior of MD in more general and challenging scenarios, we focus on its specific application for solving the penalized problem (P5) as formulated in (\ref{KLMDup}). 

Before delving into the theoretical underpinnings of the proposed MD, we will first establish fundamental definitions to characterize a proper solution to problem (P5).

\begin{defn}[First-order Stationary Point]
    For problem (P5), $\bm{X}^*\in \Delta_k^N$ is a stationary point if and only if $\left\langle \bm{X}^*(2\bm{W}-\rho \bm{I}), \bm{X}-\bm{X}^*\right\rangle\ge 0$ for all $\bm{X}\in \Delta_k^N$. 
    \label{defn:station}
\end{defn}

Then, we can derive the overall convergence analysis of the MD algorithm (\ref{alg:MD}) for solving the penalized problem P5 in the following Theorem \ref{thm:ConvMD}.

\begin{theorem}[Convergence of the MD] 
For step size $t_m=\frac{1}{L+m+1}$, the sequence $\{\bm{X}^{(m)}\}_{m=0}^{T-1}$ generated by the MD algorithm (\ref{alg:MD}) converges to a stationary point of P5 at a rate $T\ge \sqrt{\frac{2\Vert\bm{W}\Vert_{1,1}+\rho N}{{\epsilon_1}}}$ with respect to the stopping criterion $ \KL  (\bm{X}^{(m+1)},\bm{X}^{(m)}) \le {\epsilon_1}$, where $T$ is the number of iterations.
\label{thm:ConvMD}
\end{theorem}
\begin{proof}
See Appendix \ref{app:thmConvMD}.
\end{proof}

\begin{remark}
Theorem \ref{thm:ConvMD} implies that the convergence rate, as measured by the stopping criterion $ \KL  (\bm{X}^{(m+1)},\bm{X}^{(m)}) \leq {\epsilon_1}$, achieves an efficient convergence rate of $\mathcal{O}\left(\sqrt{\frac{1}{{\epsilon_1}}}\right)$. We choose the \blue{diminishing} step size $t_m=\frac{1}{L+m+1}$, \blue{where the smoothness constant $L$ accounts for the curvature of the objective in early exploration, and the $O(1/m)$ decay mitigates oscillation in later stages.}
Moreover, as the $\ell_{1,1}$-norm of matrix $\bm{W}$, the penalized parameter $\rho$, and the number of cells $N$ increase while keeping ${\epsilon_1}$ fixed, the number of iterations required for convergence will also increase. Upon obtaining MD for solving the penalized problem P5 in (\ref{opt:modk:4:condi}), we can leverage the proposed PMD method to gradually increase $\rho$  in P5 and solve the corresponding sequence of penalized problems via MD to achieve a robust solution of P4. 
\end{remark}

\subsection{Empirical Analysis of MD}

In this section, we empirically investigate the convergence rate of the MD algorithm, as detailed in Theorem \ref{thm:ConvMD}.
We synthetically construct an instance of the matrix $\bm{W}$ by first generating a random non-negative matrix $\bm{W}_0$ whose entries are i.i.d.~uniformly distributed over $[0,b]$ for some $b>0$, and then setting symmetrization $\bm{W} = (\bm{W}_0 + \bm{W}_0^T)/2$.

We consider the Min-$3$-Partition problem ($k=3$), with $b\in\{0.5, 1, 5\}$, and $N\in\{5, 20, 100\}$. We set {$\epsilon_1=10^{-4}$} and $\rho\in\{10^{-8},1\}$. Figure \ref{fig:cvgs} plots the difference between consecutive iterates $\KL(\bm{X}^{(m)},\bm{X}^{(m-1)})$ against the iteration count $m$. Each curve exhibits a general downward trend until the stopping criterion is reached. 

Specifically, in the original scenario with $N=20$, $b=1$, and $\rho=10^{-8}$, convergence is achieved in $145$ iterations. Decreasing the parameter $b$ decreases the iteration count. Similarly, decreasing $N$ leads to convergence in fewer iterations. Changes in $N$ have a more pronounced effect on the iteration count than changes in $b$, reflecting the linear convergence rate with respect to $N$ and the square root rate with respect to $b$. Increasing the penalization parameter $\rho$ slightly increases the iteration count. These trends are consistent with the analysis in Theorem \ref{thm:ConvMD}.

\begin{figure}[t]%
    \centering
    \includegraphics[width=\linewidth]{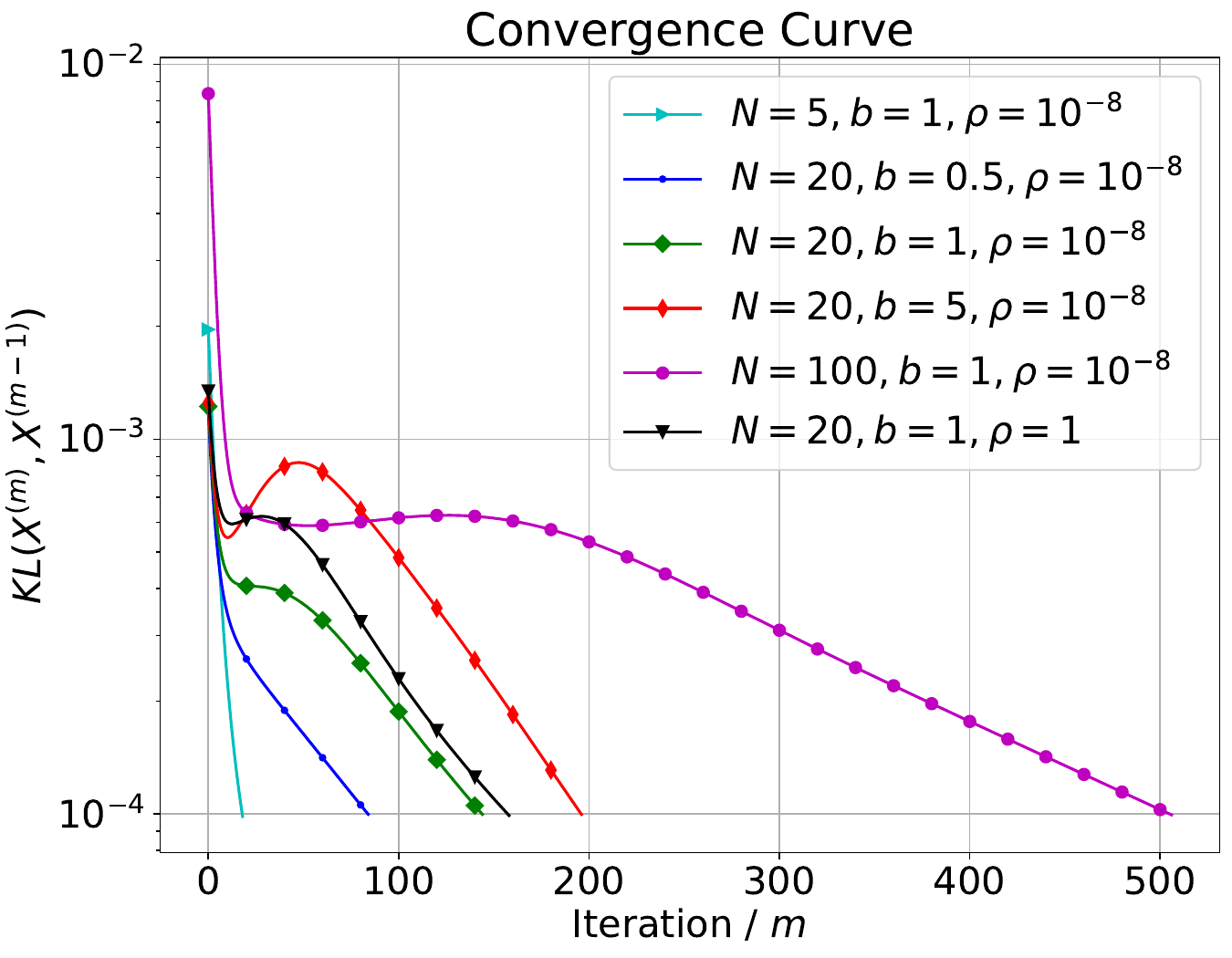}  
    \caption{Convergence curve for mod-$3$ optimization problem of MD with KL divergence.}
    \label{fig:cvgs}
\end{figure}

\subsection{Pre-processing and Post-processing of the Proposed PMD Algorithm}
To enhance the efficiency of the PMD algorithm, we further propose the following pre-processing and post-processing procedures.

\blue{In pre-processing, to ensure that the gradient of the distance-generating function is well-defined, we avoid initializing at boundary points of the simplex $\Delta_k^N$. This is achieved via a random initialization strategy for the starting point $\bm{X}^{0} = [\bm{x}^{0}_{1}, \cdots, \bm{x}^{0}_{N}]$, where each column $\bm{x}_j^0$ is generated as follows:}
\begin{equation}
\begin{aligned} \label{alg:ini}
    &\textbf{Random Initialization:} \\
    &\qquad\text{Step 1: Generate}\quad\bm{z}\overset{\text{i.i.d.s}}{\sim} U^{k\times 1}[0, 1],\\
    &\qquad \text{Step 2: Normalize}\quad\bm{x}^{0}_j= \frac{\bm{z}}{\Vert\bm{z}\Vert_1}.
\end{aligned}    
\end{equation}

In our experiments, we observe that PMD almost always returns a feasible solution to (P4), and thus a feasible solution to the Min-$k$-Partition problem. In the post-processing procedure, we propose a local search method in Algorithm \ref{alg:1} to improve the PMD solution by searching over cell pairs. It introduces a mild additional $\mathcal{O}(N^2)$ complexity.

\begin{algorithm}
\KwData{The solution obtained  by PMD $\bm{X}$}
\KwResult{The refined solution $\widetilde{\bm{X}}$ }

    \For{$i=1\text{ to }N$}{
        \For{$j=i+1\text{ to }N$}{
            \text{Calculate $(k_1^*,k_2^*) = \argmin_{(k_1,k_2)}$} \[F([\bm{x}_1,\cdots, \underbrace{\bm{e}_{k_1}}_{i\text{-th column}},\cdots,\underbrace{\bm{e}_{k_2}}_{j\text{-th column}},\cdots,\bm{x}_{N}];0)\]
            
            \text{\bf{Set}} $(\bm{x}_{i},\bm{x}_j)=(\bm{e}_{k_1^*}, \bm{e}_{k_2^*})$}
        }
    
\KwRet{$\widetilde{\bm{X}}=\bm{X}$}
\caption{Local Search}
\label{alg:1}
\end{algorithm}

\subsection{Overall Implementation of the PMD algorithm \label{sec:imple}}
The overall implementation of the PMD for relaxation-free Min-$k$-Partition is shown in Algorithm \ref{alg:3}. The PMD iteration consists of an outer loop and an inner loop, where the outer loop gradually increases the penalization parameter $\rho$ in (P5), and the inner loop solves (P5) via the update \eqref{KLMDup}. In the update \eqref{KLMDup}, we choose $t_m = 1/(L+m+1)$ according to Theorem \ref{thm:ConvMD}. The termination criterion for the outer loop is
\[
\frac{\Vert\bm{Q}\bm{X}^\tau(\bm{Q}\bm{X}^\tau)^T-\bm{I}_k\Vert_F}{k^2}\le \epsilon_2,
\]
where $\bm{Q}$ is a diagonal matrix that scales the rows of $\bm{X}^{\tau}$ so that they have unit 2-norms. It is based on the fact that $\bm{X}\in\Delta_k^N$ is feasible for (P4) only if it has orthogonal rows. This criterion is also used in \cite{wang2021clustering}.

\begin{algorithm}
\KwData{The interference matrix $\bm{W}$, the initial penalty parameter $\rho_0$, the penalty increasing parameter $\gamma$, tolerance parameters $\epsilon_1,\epsilon_2\ge 0$}
\KwResult{Solution $\widetilde {\bm{r}}$ } 

\text{}

\texttt{// Random Initialization}

\For{$j=1$ to $N$ independently }{
   \text{Generate $\bm{x}^{0}_j$ by Random Initialization (\ref{alg:ini})}
}
\text{\bf{Set}} $\bm{X}^{0}=[\bm{x}_1^0,\ldots,\bm{x}_N^0]$;

\text{}

\texttt{// PMD Iteration}

\text{\bf{Set}} $\tau=0, \rho=\rho_0$

\Repeat{$\Vert\bm{Q}\bm{X}^\tau(\bm{Q}\bm{X}^\tau)^T-\bm{I}_k\Vert_F / k^2 \le \epsilon_2$, where $\bm{Q}$ is a diagonal matrix such that the rows of $\bm{Q}\bm{X}^{\tau}$ have unit $2$-norms}
{
     \text{\bf{Set}} $m=0$, $\bm{X}^{(m)}=\bm{X}^{\tau}$\\
     \Repeat{$\KL(\bm{X}^{(m)},\bm{X}^{(m-1)})\le\epsilon_1$}{
     Update $\bm{X}^{(m+1)}$ by \eqref{KLMDup}

     \text{\bf{Set}} $m = m + 1$}

     \text{\bf{Set}} $\bm{X}^{\tau+1}=\bm{X}^{(m)}$

     \text{\bf{Set}} $\tau = \tau + 1$

     \text{\bf{Update}} $\rho=\gamma\rho$
}

\text{}

\texttt{// Local Search}

\text{\bf{Set} $\widetilde{\bm{X}}=\bm{X}^{\tau}$}

\Repeat{No further improvement}{
    
     $\widetilde{\bm{X}}= \textbf{Local Search}(\widetilde{\bm{X}})$
}

\text{{Convert} $\widetilde{\bm{X}}$ into $\widetilde{\bm{r}}$ via: $\widetilde{r}_i = j$ if $\widetilde{\bm{x}}_i = \bm{e}_j$, $j=1,...,k$}.

\KwRet{$\widetilde {\bm{r}}$}
\caption{Proposed PMD Method for Relaxation-free Min-$k$-Partition}
\label{alg:3}
\end{algorithm}

\section{Extension: \blue{Optimizing PCIs} of a Subset of Cells \label{sec:ext}}

In our discussion so far, we have taken the approach of assigning the PCI for a given network from scratch. In practice, however, practitioners usually have a (possibly handcrafted and imperfect) PCI assignment beforehand and seek to improve it by making changes to the PCIs of only \emph{a subset} of the cells. Such a partial update can be more desirable than directly updating the PCI of the entire network because the latter may drastically change the stability of the whole 5G network. In this section, we show that our graph partitioning approach can be naturally extended to this setting.

Mathematically, we are given an existing PCI assignment $\textbf{PCI}^{(0)}=(\text{PCI}_1^{(0)},...,\text{PCI}_N^{(0)})$ and a subset of \emph{changeable cells} $\mathcal{S}\subseteq\mathcal{V}$. Let $\mathcal{U}=\mathcal{V}\backslash\mathcal{S}$, which denotes the set of \emph{unchangeable cells}. Our goal is to optimize mod-$3$ and mod-$30$ interference and eliminate collisions and confusions, by only updating the PCIs of the changeable cells.

To adapt our approach to this setting, we apply the decomposition framework in Section \ref{sec:decomposition} to only the changeable cells, where the PCI assignment $(\text{PCI}_i)_{i\in\mathcal{S}}$ for the changeable cells is decomposed into mod-$3$ value assignment $(\widetilde{r}_i^{(3)})_{i\in\mathcal{S}}$, mod-$10$ value assignment $(\widetilde{r}_i^{(10)})_{i\in\mathcal{S}}$, and quotient assignment $(q_i)_{i\in\mathcal{S}}$. At the same time, we need to take into account the interaction between changeable and unchangeable cells. We now explain how this can be achieved.

\subsection{Mod-$3$ Value and Mod-$10$ Value Assignment}

To assign mod-$3$ and mod-$10$ values (sub-problems 1.1 and 1.2 in Section \ref{sec:mod-3 and mod-10}), we consider a variant of the Min-$k$-Partition problem where the optimization variables are restricted to the changeable cells:
\begin{align*}
    \min_{(r_i)_{i\in\mathcal{S}}} &\quad \sum_{i,j\in\mathcal{S}} [\bm{W}]_{i,j} \cdot \mathbbm{1} \{ r_i = r_j \} \\ &\quad + 2 \sum_{i\in\mathcal{S},\,j\in\mathcal{U}} [\bm{W}]_{i,j} \cdot \mathbbm{1} \{ r_i = r_j\} \\ &\quad +
    \sum_{i,j\in\mathcal{U}} [\bm{W}]_{i,j} \cdot \mathbbm{1} \{ r_i = r_j\}, \\
    \text{s.t.}&\quad r_i\in\mathbb{Z}_k,\quad\forall i\in\mathcal{S}.
\end{align*}
Here $(r_j)_{j\in\mathcal{U}}$ are pre-determined from the PCIs of the unchangeable cells. Moreover, the last term in the objective function corresponds to the interference between unchangeable cells and is thus a constant that can be ignored. 

To modify the PMD algorithm for Min-$k$-Partition in Section \ref{sec:prop}, we introduce block matrix notations $\bm{W}_{\mathcal{S},\mathcal{S}}=[\bm{W}_{i,j}]_{i\in\mathcal{S},j\in\mathcal{S}}$ and $\bm{W}_{\mathcal{S},\mathcal{U}}=[\bm{W}_{i,j}]_{i\in\mathcal{S},j\in\mathcal{U}}$. Then the problem (P5) for the PMD algorithm becomes
\begin{align*}
    \min_{\bm{X}_{\mathcal{S}}}&\quad \text{Tr}\left(\bm{X}_{\mathcal{S}} \left(\bm{W}_{\mathcal{S},\mathcal{S}} - \frac{\rho}{2} \bm{I} \right) \bm{X}_{\mathcal{S}}^T\right) + 2\text{Tr}\left(\bm{X}_{\mathcal{S}}\bm{W}_{\mathcal{S},\mathcal{U}}\bm{X}_{\mathcal{U}}^T \right) \\
    \text{s.t.} &\quad \bm{X}_{\mathcal{S}} \in \Delta_k^{|\mathcal{S}|},
\end{align*}
where $\bm{X}_{\mathcal{U}}\in\{\bm{e}_1,...,\bm{e}_k\}^{|\mathcal{U}|}$ can be deduced from the existing PCI assignment. The corresponding PMD algorithm can then be easily derived.

\subsection{Quotient Assignment}

Finally, to assign quotient values (Section \ref{sec:quotient}), we solve a variant of the graph coloring problem where an existing coloring (quotient assignment) is given, and the goal is to update the colors (quotients) of changeable cells to minimize collisions and confusions. Since we cannot eliminate the collisions and confusions between two unchangeable cells, we will focus only on cell pairs that have at least one changeable cells. Thus, the problem can be formulated as
\begin{equation}
    \begin{aligned}
    \text{find} &\quad q_i,\quad i\in\mathcal{S}, \\
    \text{s.t.} &\quad q_i\neq q_j,\quad\forall (i,j) \in \mathcal{E} \cap \big[(\mathcal{S}\times\mathcal{S}) \cup (\mathcal{S}\times\mathcal{U})\big], \\
    &\quad q_i\in\mathbb{Z}_+,\quad \forall i\in\mathcal{S}.
\end{aligned}
\label{eqn-partial coloring}
\end{equation}
Here again, $(q_j)_{j\in\mathcal{U}}$ can be computed from the existing PCI assignment, and we ignore the PCI range constraint as in Section \ref{sec:quotient}. We have also slightly abused notation, as quotient assignment is done over each individual mod-$30$ cluster instead of the whole set of changeable cells.

We now show that the problem \eqref{eqn-partial coloring} can, in fact, be reduced to the graph coloring problem. Let $c_1,...,c_A$ be the colors (quotients) used by the unchangeable cells. We will replace the unchangeable cells with a clique consisting of $A$ vertices $v_1,...,v_A$, where $v_a$ represents the used color $c_a$. Then we add an edge between a changeable cell $i$ and $v_a$ whenever some unchangeable cell $j$ has used the color $c_a$ and $(i,j)\in\mathcal{E}$. 
This leads to a graph $\mathcal{H}$ with vertices $\mathcal{S}\cup\{v_1,...,v_A\}$ and edges $\mathcal{F}=\mathcal{F}_1\cup\mathcal{F}_2\cup\mathcal{F}_3$, where
\begin{align*}
    & \mathcal{F}_1 = \mathcal{E} \cap (\mathcal{S} \times \mathcal{S}), \\
    & \mathcal{F}_2 = \{(i,v_a):i\in\mathcal{S}, \text{ and } \exists j\in\mathcal{U} \text{ s.t.~} (i,j)\in\mathcal{E},\, q_j=c_a\}, \\
    & \mathcal{F}_3 = \{(v_a,v_{a'}) : a\neq a'\}.
\end{align*}

We will now see that performing graph coloring over this new graph $\mathcal{H}$ yields a solution to \eqref{eqn-partial coloring}. Thanks to the clique structure $\mathcal{F}_3$, the vertices $v_1,...,v_A$ will receive distinct colors, which can be made $c_1,...,c_A$ after relabeling. Moreover, $\mathcal{F}_1$ and $\mathcal{F}_2$ force the color assignment to eliminate collisions and confusions between pairs $(i,j)\in\mathcal{S}\times\mathcal{S}$, and between pairs $(i,j)\in(\mathcal{S},\mathcal{U})$, respectively. Therefore, this gives a solution to \eqref{eqn-partial coloring}.

\section{\blue{Numerical} Experiments}\label{sec:exp}

This section presents experiments to demonstrate the performance and efficiency of our PCI assignment approach compared with state-of-the-art baselines. All experiments are implemented in Python on a machine equipped with an Intel i9-12900K core.

\subsection{List of Baseline Methods}

We call our proposed method of graph partitioning with PMD by \emph{GP-PMD}. The baseline methods for comparison are described as follows. We select three representative methods from three main categories for solving the PCI problem.

\textbf{Greedy Graph Coloring Method (GGC)}~\cite{bandh2009graph}: This method employs a greedy coloring algorithm to address collisions and confusions. In the simulation, we utilize the codes provided from \cite{coloring}.

\textbf{Greedy Search for BQP (BQP)}~\cite{gui2018pci2}: We employed BQP methods to handle collisions, confusions, mod-$3$ interference, and mod-$30$ interference. We add the mod-$30$ interference into the objective function and set the hyperparameters to those used in the experiment {\em 1-1-1  Greedy} in  \cite{gui2018pci2}.

\textbf{Genetic Method (Genetic)}~\cite{panxing2016pci}: We apply a genetic method for comparison. The hyperparameter settings include a population size of $30$ and $200$ iterations, and the probabilities of crossover, mutation, and selection are set to $0.77$, $0.3$, and $0.8$, respectively.

\blue{
\textbf{Memetic Algorithm (MA-NWS)}~\cite{andrade2022physical}: We adopt the non-warm-start version as recommended by~\cite{andrade2022physical}, which relies solely on the core evolutionary mechanism of the Memetic Algorithm. All hyperparameters are set according to the best-performing configuration reported in~\cite{andrade2015biased}.
} 

\textbf{Graph Partitioning with Penalized Gradient Projection (GP-PGP)}: To evaluate the performance of the PMD algorithm, we consider an alternative benchmark that uses the same graph partitioning framework as GP-PMD but applies penalized gradient projection (PGP) to solve sub-problems 1.1 and 1.2. \blue{It also uses the local search procedure in Algorithm \ref{alg:1} to improve the PGP solution. The stopping criterion for GP-PGP is set by the Euclidean Bregman divergence: $\frac{1}{2}\Vert\bm{X}^{(m+1)} - \bm{X}^{(m)}\Vert_F^2 \le \epsilon_1$. All hyperparameters are the same as those used in GP-PMD.

\textbf{Graph Partitioning with Local Search (GP-LS)}: To evaluate the performance of the PMD algorithm, we include an additional baseline. It uses the same graph partitioning framework as GP-PMD, but directly generates a random solution when solving sub-problems 1.1 and 1.2, and then applies the local search procedure in Algorithm \ref{alg:1}.
}

\blue{

\subsection{Synthetic Data Experiments}

In this section, we test our method on synthetically generated networks.

\subsubsection{Experiment Setup}

\begin{table}[t]
    \centering
    \caption{Structural statistics of the random geometric graphs used in our experiments (averaged over 30 instances per setting).}
    \begin{tabular}{ccccc}
        \midrule
        $(N, r)$ & Density & Max Clique & Avg. Degree & Clustering Coef. \\
        \midrule
        $(100,\ 0.1)$ & $0.0279$ & $5.00$  & $5.075$   & $0.719$ \\
        $(100,\ 0.2)$ & $0.1044$ & $9.67$  & $26.333$  & $0.717$ \\
        $(100,\ 0.3)$ & $0.2141$ & $15.17$ & $54.299$  & $0.787$ \\
        $(200,\ 0.1)$ & $0.0285$ & $7.43$  & $13.050$  & $0.706$ \\
        $(200,\ 0.2)$ & $0.1049$ & $15.40$ & $58.734$  & $0.715$ \\
        $(200,\ 0.3)$ & $0.2156$ & $25.90$ & $114.909$ & $0.796$ \\
        $(500,\ 0.1)$ & $0.0289$ & $13.43$ & $40.993$  & $0.655$ \\
        $(500,\ 0.2)$ & $0.1055$ & $30.47$ & $159.088$ & $0.721$ \\
        $(500,\ 0.3)$ & $0.2159$ & $55.40$ & $298.456$ & $0.802$ \\
        \midrule
    \end{tabular}
    \label{tab:rgg_stats}
\end{table}

We generate synthetic networks using the random geometric graph model $\mathrm{RGG}(N, r)$, where $N$ denotes the number of cells and $r$ is called the connection radius. Cells are uniformly distributed in a 2D unit square, and cell $j$ is treated as a neighbor of node $i$ if their Euclidean distance is at most $r$. The interference matrix $\bm{W}$ is defined by setting the interference inversely proportional to the Euclidean distance, simulating spatial interference in wireless networks. This model is widely used to emulate realistic wireless topologies~\cite{Gowaikar2006, Haenggi2009}.

To evaluate algorithm performance under diverse topological conditions, we systematically vary both the network size and spatial density. Specifically, we consider network sizes $N \in \{100, 200, 500\}$, and connection radii $r \in \{0.1, 0.2, 0.3\}$ which correspond to sparse, medium, and dense graphs, respectively. For each configuration $\mathrm{RGG}(N, r)$, we generate 30 independent instances, which yields a total of 270 graph instances. For each type of graph, we compute standard structural metrics, including edge density, maximum clique size, average degree, and clustering coefficient. These statistics are summarized in Table \ref{tab:rgg_stats} and allow us to interpret algorithm performance under varying structural complexities. 
The hyperparameters of GP-PMD is set as $\rho_0=10^{-8}$, $\gamma=1.1$, $\epsilon_1=10^{-5}$, and $\epsilon_2=10^{-10}$, which is the same configuration as~\cite{wang2021clustering}. We also conduct an ablation study to examine the effects of the hyperparameters on the performance of the algorithm, and the results are provided in Section B of the supplementary material~\cite{qiu2025relaxation}.

\subsubsection{Experiment Results}

\begin{figure*}[t]
    \centering
    \includegraphics[width=\linewidth]{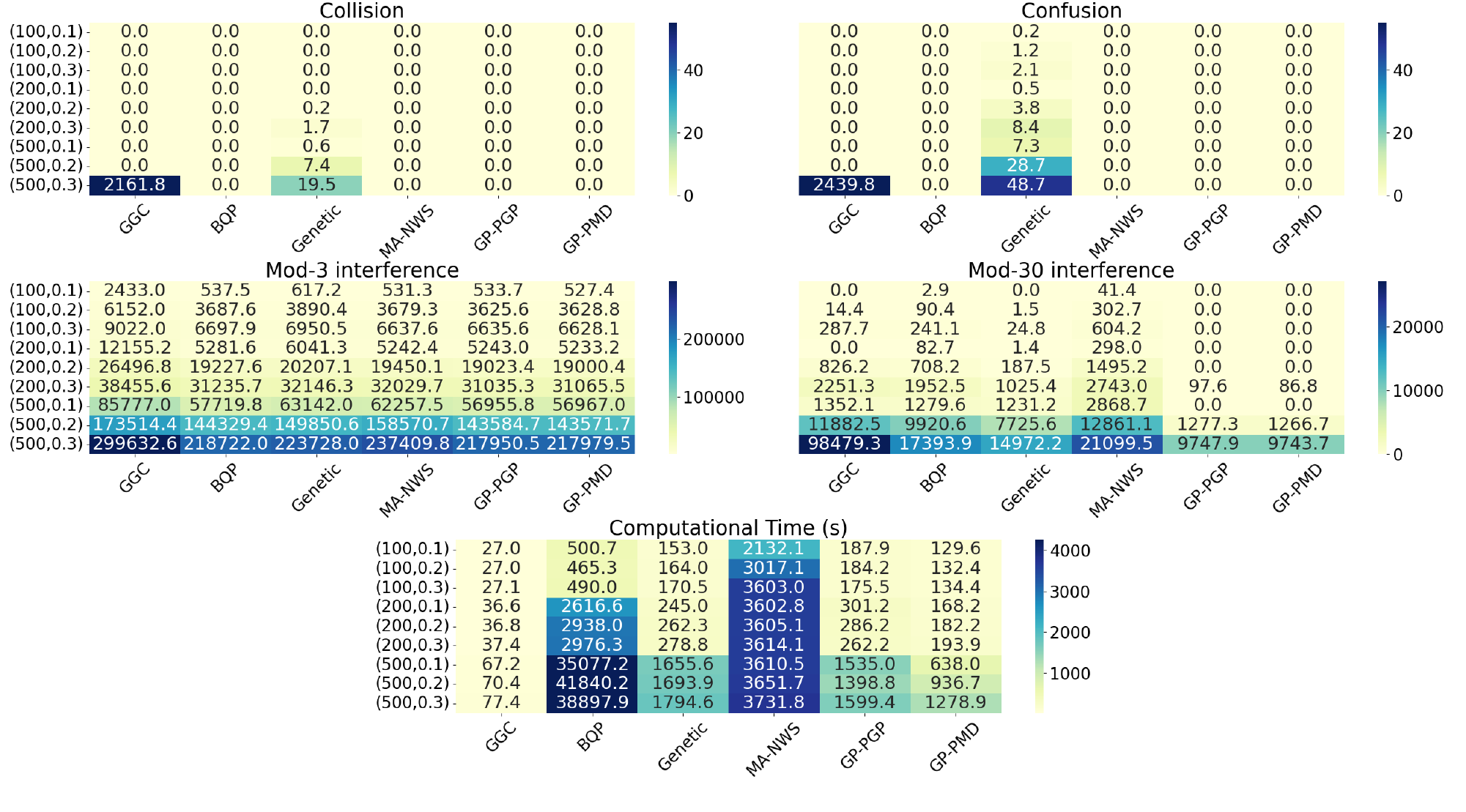}
    \caption{Experiment results on random geometric graphs. 
    }
    \label{fig:heatmap}
\end{figure*}

We compare all baseline algorithms on the random geometric graphs using five evaluation metrics: mod-$3$ interference, mod-$30$ interference, number of collisions, number of confusions, and computational time. The overall performance is visualized in a heatmap (Fig.~\ref{fig:heatmap}). The detailed box-plots are included in Section C of the supplementary material~\cite{qiu2025relaxation}. The key observations are summarized as follows.

GGC is the fastest among all methods (under \SI{80}{s} across all instances), but \blue{fails to reduce} mod-$3$ or mod-$30$ interference. 
BQP eliminates all collisions and confusions across all test instances, and its interference control is generally competitive, though slightly behind MA-NWS on small and medium-sized graphs. However, its computational cost escalates as the graph size increases. 
Genetic achieves moderate performance across all metrics. It reduces mod-$3$ and mod-$30$ interference but fails to eliminate collisions and confusions. Compared to our proposed methods, it exhibits weaker control over mod-$3$ and mod-$30$ interference. While its runtime is substantially shorter than that of BQP and MA-NWS, it remains noticeably longer than ours.
MA-NWS successfully eliminates collisions and confusions and achieves comparable performance on mod-$3$ interference compared with our methods on smaller instances. However, it performs worse than our methods on mod-$30$ interference and has a relatively high runtime. 
GP-PGP and GP-PMD consistently achieve strong performance: both eliminate all collisions and confusions and achieve the lowest interference on mod-$3$ and mod-$30$. While both approaches yield comparable solution quality, GP-PMD is more efficient and requires less computational time. 
Overall, the GP-PMD achieves a compelling trade-off between feasibility, interference suppression, and scalability, making it the most robust and versatile algorithm among the compared baselines. 
Moreover, to assess the statistical significance of our results, we conduct statistical hypothesis tests in Section D of the supplementary materials~\cite{qiu2025relaxation}. 

}

\subsection{\blue{Real-Data Experiments}}

\subsubsection{Experiment Setup}

We \blue{test our method on} real-world data collected in Beijing. The dataset features antennas deployed in a SISO system, focusing on scenarios involving mod-$3$ interference, mod-$30$ interference, collisions, and confusions in simulations.
The raw MR data was extracted from a live downlink network comprising $10,000$ cells, of which 673 are changeable, covering a dense urban area of $53.0859$ square kilometers. This data, gathered through regular measurements by UEs reporting network status and performance, was utilized to generate the sets $\mathcal{E}_1$ and $\mathcal{E}_2$, and the matrix $\bm{W}$.

For collision detection, we include the pair $(i,j)\in\mathcal{E}_1$ if both cells share the same frequency and MR data indicates that cell $j$ is a neighbor of cell $i$. 
In the case of confusion, we include the pair $(i,j)\in\mathcal{E}_2$ if they share the same frequency and there exists another cell $\ell$ for which MR data shows that $i$ and $j$ are both its neighbors. 
For the interference matrix $\bm{W}$, if cell $j$ is a neighboring cell of cell $i$ and shares the same frequency as cell $i$, the entry $[\bm{W}]_{i,j}$ reflects the number of MRs between these two cells.
Furthermore, to assess performance across different \blue{network sizes}, we \blue{select subsets of the main dataset with} varying numbers of cells $N$ and changeable cells $S=|\mathcal{S}|$.
\blue{These subsets consist of cells with consecutive cell IDs,} and have $N$ values of $401$, $1398$, $4396$, and $10,000$, with corresponding $S$ values of $136$, $231$, $368$, and $673$. 
These choices are designed to assess the algorithm's performance across a range of problem scales.

\subsubsection{Experiment Results}

We summarize the experiment results in Table \ref{res:n401}, \ref{res:n1398}, \ref{res:n4396}, and \ref{res:n10000}. \footnote{In these tables,`Coll.' indicates the number of collisions, `Conf.' refers to the number of confusions, `Mod-$3$' denotes the mod-$3$ interference value, 'Mod-$30$' is the mod-$30$ interference value, and 'Time' represents the runtime of the methods. }
The results demonstrate that GP-PMD significantly outperforms all baseline methods in minimizing mod-$3$ and mod-$30$ interference, and resolving collisions and confusions, all while taking the least computational time. 
Specifically, compared to the state-of-the-art methods, GP-PMD reduces mod-$3$ interference by $8.3\%$ and mod-$30$ interference by $99.7\%$, and achieves a $20$ times speedup in runtime on average.
In contrast, GGC fails to address the minimization of mod-$3$ and mod-$30$ interference. Genetic fails to achieve collision-free and confusion-free solutions. Genetic and BQP all produce significantly higher mod-$3$ and mod-$30$ interference and require substantially longer computational time. Moreover, BQP is unable to generate any solution within two weeks for large-scale instances. 
Lastly, the comparison with GP-PGP highlights the superior efficiency and effectiveness of the proposed PMD algorithm in minimizing mod-$3$ and mod-$30$ interference.

\begin{table}[t]
\centering
\begin{threeparttable}
\caption{The evaluation results on the real-world data with $N=401, S=136$.}
\begin{tabular}{c|ccccc}
\hline
Models                   & Coll. & Conf. & Mod-3. & Mod-30. & Time (s) \\ \hline
GGC \cite{coloring, bandh2009graph}           & 0     & 0     & 2887700 & 614093 &  48      \\ \hline
BQP \cite{gui2018pci2} & 0     & 0     & 668280 & 15587 & 2297  \\ 
Genetic \cite{panxing2016pci} & 2     & 2     & 758546 & 17745 & 376   \\
MA-NWS \cite{andrade2022physical} & 0     & 0     & 640415 & 55505 & 422  \\ \hline
GP-LS & 0 &0&690063&1&179\\
GP-PGP & 0     & 0     & 641961  & \textbf{0} & 279 \\
GP-PMD & 0     & 0     & \textbf{628418}  & \textbf{0} & 197  \\ \hline
\end{tabular}
\label{res:n401}
\end{threeparttable}
\end{table}

\begin{table}[t]
\centering
\begin{threeparttable}
\caption{The evaluation results on the real-world data with $N=1398, S=231$.}
\begin{tabular}{c|ccccc}
\hline
Models                   & Coll. & Conf. & Mod-3. & Mod-30. & Time (s) \\ \hline
GGC~\cite{coloring, bandh2009graph}           & 0     & 0     & 4283366 & 429858 &    85   \\ \hline
BQP~\cite{gui2018pci2} & 0     & 0     & 1216883 & 39039 & 18936  \\
Genetic~\cite{panxing2016pci} & 2     & 5     & 1592459 & 32406 & 1407   \\ 
MA-NWS~\cite{andrade2022physical} & 0     & 16     & 1236074 & 104253 & 731  \\ \hline
GP-LS & 0 &0&1207612&326&332\\
GP-PGP & 0     & 0     & 1185974  &  \textbf{159} & 372 \\
GP-PMD & 0     & 0     & \textbf{1143088} & 170 & 352   \\ \hline
\end{tabular}
\label{res:n1398}
\end{threeparttable}
\end{table}

\begin{table}[t]
\centering
\begin{threeparttable}
\caption{The evaluation results on the real-world data with $N=4396, S=368$.  }
\begin{tabular}{c|ccccc}
\hline
Models                   & Coll. & Conf. & Mod-3. & Mod-30. & Time (s) \\ \hline
GGC~\cite{coloring, bandh2009graph} & 0     & 0     & 9349226 & 753756 & 107   \\ \hline
BQP~\cite{gui2018pci2}${}^*$ & -     & -     & - & -  &  -  \\ 
Genetic~\cite{panxing2016pci} & 14     & 16     & 5159102 & 140707 & 7403   \\
MA-NWS~\cite{andrade2022physical} & 0     & 25     & 2909566 & 255494 & 2000  \\ \hline
GP-LS & 0 &0&3429811&495&1099\\
GP-PGP & 0     & 0     & 2833178  & 287 & 1179 \\
GP-PMD & 0     & 0     & \textbf{2831831} & \textbf{216}  & 971   \\ \hline
\multicolumn{6}{l}{\small $*$: BQP fails to return a solution in two weeks of time.}\\
\end{tabular}
\label{res:n4396}
\end{threeparttable}
\end{table}

\begin{table}[t]
\centering
\begin{threeparttable}
\caption{The evaluation results on the real-world data with $N=10000, S=673$.  }
\begin{tabular}{c|ccccc}
\hline
Models                   & Coll. & Conf. & Mod-3. & Mod-30. & Time (s) \\ \hline
GGC~\cite{coloring, bandh2009graph}  & 0     & 0     & 18039785  & 1052895 & 254  \\ \hline
BQP~\cite{gui2018pci2}${}^*$ & -     & -     & - & -  & -  \\ 
Genetic~\cite{panxing2016pci} & 33    & 32 & 11642161 & 574454 & 55515   \\
MA-NWS~\cite{andrade2022physical} & 0     & 91     & 6838767 & 580523 & 12398  \\ \hline
GP-LS & 0 & 0     &7466815&19367    &6967\\
GP-PGP & 0     & 0     & 6607312 & \textbf{6451}  & 6672 \\
GP-PMD & 0     & 0     & \textbf{6584238} &  6617  & 6185   \\ \hline
\multicolumn{6}{l}{\small $*$: BQP fails to return a solution in two weeks of time.}\\
\end{tabular}
\label{res:n10000}
\end{threeparttable}
\end{table}

{

\section{Conclusion}\label{sec:con}

In this paper, we introduce a novel approach for PCI assignment, decomposing it into Min-$3$-Partition, Min-$10$-Partition, and graph coloring. To solve the challenging NP-hard Min-$k$-Partition problem, we present a relaxation-free reformulation that can be solved by an efficient PMD algorithm. Theoretical results are developed to provide strong guarantees for the proposed method.
Experimental evaluations on synthetic and real-world datasets show the superior performance and computational efficiency of the proposed method over baseline methods. 
\blue{While our method shows strong performance on networks with up to $10,000$ cells, scaling to even larger networks may benefit from more advanced techniques. Future work could explore neural optimization or learning-based surrogates to further improve efficiency and generalization in interference-aware settings.}

\appendices

\blue{
\section{Proof of Lemma~\ref{prop:1} \label{proof:prop:1}}

\begin{proof}
    Clearly, if $\bm{x}$ has one non-zero entry, then $\|\bm{x}\|_p = \|\bm{x}\|_q$. Below we prove that if 
$\Vert \bm{x}\Vert_q= \Vert\bm{x}\Vert_p$, then $\bm{x}$ has exactly one non-zero entry.

First, we have
\begin{align}
    \Vert \bm{x}\Vert_q^q &= \sum_{i=1}^k\vert x_i\vert^q = \sum_{i=1}^k\vert x_i\vert^{p}\vert x_i\vert^{q-p} \notag\\
    &\overset{(a)}{\le} \left(\sum_{i=1}^k\vert x_i\vert^p\right)\left(\sum_{i=1}^k\vert x_i\vert^{q-p}\right), \label{app:ineqn1}
\end{align}
and
\begin{align}
    \Vert \bm{x}\Vert_q^q &\overset{(b)}{=} \Vert \bm{x}\Vert_p^q = \left(\sum_{i=1}^k\vert x_i\vert^p\right)^{\frac{q}{p}} \notag\\
    &= \left(\sum_{i=1}^k\vert x_i\vert^p\right)\left(\sum_{i=1}^k\vert x_i\vert^p\right)^{\frac{q-p}{p}} \notag\\
    &\overset{(c)}{\ge} \left(\sum_{i=1}^k\vert x_i\vert^p\right)\left(\sum_{i=1}^k\vert x_i\vert^{p\cdot \frac{q-p}{p}}\right) \notag\\
    &= \left(\sum_{i=1}^k\vert x_i\vert^p\right)\left(\sum_{i=1}^k\vert x_i\vert^{q-p}\right). \label{app:ineqn2}
\end{align}

Here, (a) is due to the inequality $\sum_{i=1}^k a_ib_i \le (\sum_{i=1}^k a_i)(\sum_{i=1}^k b_i)$ for $a_i, b_i \ge 0$,  
(b) follows from the assumption in the necessary condition, and  
(c) uses the inequality $\sum_{i=1}^k a_i^t \le (\sum_{i=1}^k a_i)^t$ for $a_i \ge 0$ and $t>0$.

Combining (\ref{app:ineqn1}) and (\ref{app:ineqn2}), we have
\begin{align*}
    \Vert \bm{x}\Vert_q^q = \left(\sum_{i=1}^k\vert x_i\vert^p\right)\left(\sum_{i=1}^k\vert x_i\vert^{q-p}\right).
\end{align*}
On the other hand,
\begin{align*}
    \Vert \bm{x}\Vert_q^q = \sum_{i=1}^k\vert x_i\vert^p\vert x_i\vert^{q-p}.
\end{align*}
We thus obtain
\begin{align*}
    \sum_{i\neq j}\vert x_i\vert^p\vert x_j\vert^{q-p} = 0,
\end{align*}
which implies
\begin{align}
    \vert x_i\vert^p\vert x_j\vert^{q-p} = 0, \quad \forall i \ne j. \label{app:eqn1}
\end{align}

Since $\bm{x} \ne \bm{0}$, there exists an index $i^*$ such that $\vert x_{i^*} \vert^p > 0$.  
From (\ref{app:eqn1}), it follows that
\begin{align}
    \vert x_j\vert^{q-p} = 0, \quad \forall j \ne i^*, \label{app:eqn2}
\end{align}
which implies $x_j = 0$ for all $j \ne i^*$.  
Hence, $\bm{x}$ has exactly one non-zero entry.
\end{proof}

\section{Derivation of (P5) \label{app:0.75}}

We substitute $p=1, q=2,v=2$ into (\ref{obj:penalty}) to obtain
\begin{align}
    &\quad\sum_{i=1}^N\sum_{j=1}^N[\bm{W}]_{i,j}\bm{x}_i^T\bm{x}_j + \frac{\rho}{2}\sum_{i=1}^N(\Vert \bm{x}_i\Vert_1^2-\Vert\bm{x}_i\Vert_2^2) \notag\\
    &=\text{Tr}(\bm{X}\bm{W}\bm{X}^T)+\frac{\rho}{2}\sum_{i=1}^N(1-\Vert\bm{x}_i\Vert_2^2) \notag\\
    &=\text{Tr}(\bm{X}\bm{W}\bm{X}^T)-\frac{\rho}{2}\text{Tr}(\bm{X}\bm{X}^T)+\frac{\rho}{2}N \notag \\
    &=\text{Tr}\left(\bm{X} \left(\bm{W}-\frac{\rho}{2}\bm{I} \right) \bm{X}^T \right)+\frac{\rho}{2}N.
\end{align}

}

\section{Proof of Theorem \ref{theo:1}\label{app:1}}

We first prove the strictly concavity of $F(\bm{X};\rho)$ when $\rho>\rho_0=2\max(\lambda_{\max}(\bm{W}), 0)$.

\begin{lemma}[Strict Concavity of $F$]
    Let $\rho_0=\max(2\lambda_{\max}(\bm{W}), 0)$. When $\rho>\rho_0$, $F(\bm{X};\rho)$ is a strictly concave function of $\bm{X}$.
    \label{concave}
\end{lemma}
\begin{proof}
Let the eigenvalue decomposition of the interference $\bm{W}$ be given by $\bm{W}=\bm{P}\bm{\Sigma}\bm{P}^T$, where $\bm{P}$ is an orthogonal matrix, $\Sigma$ is a diagonal matrix whose diagonal entries are the eigenvalues of $\bm{W}$. Then
\[
    \bm{W}-\frac{\rho}{2}\bm{I} = \bm{P}\bm{\Sigma}\bm{P}^T-\frac{\rho}{2}\bm{P}\bm{P}^T
    =\bm{P}\left(\Sigma-\frac{\rho}{2}\bm{I}\right)\bm{P}^T,
\]
Since $\rho>\rho_0=2\lambda_{\max}(\bm{W})$, we have the diagonal elements of $\bm{\Sigma}-\frac{\rho}{2}\bm{I}$ are all strictly less than $0$, making $\bm{W}-\frac{\rho}{2}\bm{I}$ a negative definite matrix. Hence, $F(\bm{X};\rho)$ is strictly concave.     
\end{proof}

A useful fact we will use is that every $\bm{X}\in\Delta_k^N$ can be written as a convex combination of matrices where each column vector only has one non-zero element. Specifically, we express $\bm{X}=[\bm{x}_1,\cdots,\bm{x}_N]\in\Delta_k^N$ as
\begin{align}
    \bm{X}=& \sum_{i_1=1}^k [\bm{X}]_{i_1,1}[\bm{e}_{i_1}, \bm{x}_{2}, \cdots, \bm{x}_{N}]\notag\\
    =&\sum_{i_1=1}^k\sum_{i_2=1}^k [\bm{X}]_{i_1,1}[\bm{X}]_{i_2,2}[\bm{e}_{i_1},\bm{e}_{i_2},\cdots,\bm{x}_{N}] \notag\\ 
    &\cdots\notag \\
    =&\sum_{i_1=1}^k\cdots\sum_{i_N=1}^k \left(\prod_{j=1}^N[\bm{X}]_{i_j,j}\right)[\bm{e}_{i_1},\cdots,\bm{e}_{i_N}].\label{rewrite}
\end{align}
Moreover, for all $\bm{X}\in\Delta_k^N$,
\begin{equation}\label{for:sum1}
    \sum_{i_1=1}^k\cdots\sum_{i_N=1}^k\prod_{j=1}^N[\bm{X}]_{i_j,j}
    = \prod_{j=1}^N \sum_{i=1}^k [\bm{X}]_{i,j} = \prod_{j=1}^N 1 = 1.
\end{equation}
Combining (\ref{rewrite}) and (\ref{for:sum1}), we conclude that every $\bm{X}\in\Delta_k^N$ can be expressed as a convex combination of matrices where each column vector has only one non-zero entry.

We are now ready to prove Theorem \ref{theo:1}. To this end, we will that every optimal solution to (P4) is optimal for (P5), and vice versa. We denote the feasible regions of (P4) and (P5) by $\mathcal{X}$ and $\mathcal{X}_{\rho}$, respectively.

First, suppose $\bm{X}^*$ is an optimal solution to (P4). Clearly $\bm{X}^*$ is a feasible solution to (P5). By the concavity of $F(\bm{X};\rho)$ from Lemma \ref{concave} and (\ref{rewrite}), for every feasible solution $\bm{X}$ to (P5),
\begin{align}
&F(\bm{X};\rho) \notag \\
&\overset{(a)}{=}F\left(\sum_{i_1=1}^k\cdots\sum_{i_N=1}^k \Big(\prod_{j=1}^N[\bm{X}]_{i_j,j}\Big)[\bm{e}_{i_1},\cdots,\bm{e}_{i_N}];\rho\right) \notag\\
&\overset{(b)}{\ge} \sum_{i_1=1}^k\cdots\sum_{i_N=1}^k \left(\prod_{j=1}^N[\bm{X}]_{i_j,j}\right) F([\bm{e}_{i_1},\cdots,\bm{e}_{i_N}];\rho)\notag \\
&\overset{(c)}{\ge} \sum_{i_1=1}^k\cdots\sum_{i_N=1}^k\left(\prod_{j=1}^N[\bm{X}]_{i_j,j}\right)\min_{\bm{X}\in\mathcal{X}} F(\bm{X};\rho)\notag\\
&\overset{(d)}{=} \sum_{i_1=1}^k\cdots\sum_{i_N=1}^k \left(\prod_{j=1}^N[\bm{X}]_{i_j,j}\right) \min_{\bm{X}\in\mathcal{X}}F(\bm{X};0)\notag\\
&\overset{(e)}{=}\sum_{i_1=1}^k\cdots\sum_{i_N=1}^k\left(\prod_{j=1}^N[\bm{X}]_{i_j,j}\right)F(\bm{X}^*;0)\notag\\
&\overset{(f)}{=}F(\bm{X}^*;0)\overset{(d)}{=}F(\bm{X}^*;\rho). \label{for:P4}
\end{align}
Here (a) follows from (\ref{rewrite}); (b) is due to the concavity of $F(\bm{X};\rho)$ from Lemma \ref{concave}; (c) holds because $[\bm{e}_{i_1},\cdots,\bm{e}_{i_N}]\in\mathcal{X}$; (d) is due to $F(\bm{X};\rho)=F(\bm{X};0)$ for all $\bm{X}\in\mathcal{X}$ and $\rho\ge 0$, as the penalized term $\sum_{i=1}^N(\Vert \bm{x}_i \Vert_1^2-\Vert\bm{x}_i\Vert_2^2)$ becomes $0$; (e) uses the global optimality of $\bm{X}_*$ in problem (P4); (f) follows from (\ref{for:sum1}).
Thus, (\ref{for:P4}) shows that every optimal solution to (P4) is optimal for (P5).

Now suppose that $\bm{X}^*$ is an optimal solution to (P5). We note that $\bm{X}^*$ must be a feasible solution to (P4). Otherwise there would exist $i_1,...,i_N$ such that $0<\prod_{j=1}^N[\bm{X}^*]_{i_j,j}<1$. Then,
\begin{align}
        &F(\bm{X}^*;\rho) \notag\\
        &\overset{(g)}{>}\sum_{i_1=1}^k\cdots\sum_{i_N=1}^k \left(\prod_{j=1}^N[\bm{X}^*]_{i_j,j}\right) F([\bm{e}_{i_1},\cdots,\bm{e}_{i_N}];\rho)\notag\\
        &\ge\sum_{i_1=1}^k\cdots\sum_{i_N=1}^k \left(\prod_{j=1}^N[\bm{X}^*]_{i_j,j}\right)\min_{\bm{X}\in\mathcal{X}_{\rho}}F(\bm{X};\rho)\notag\\
        &\overset{(h)}{=}\min_{\bm{X}\in\mathcal{X}_{\rho}}F(\bm{X};\rho),
    \end{align}
    where (g) follows from the strictly concave of $F$ from Lemma \ref{concave}, and (h) is due to (\ref{for:sum1}). This contradicts the optimality of $\bm{X}^*$. Thus, $\bm{X}^*$ is a feasible solution to (P4). (A quicker proof: every strictly concave function must attain optimum at extreme points of the convex domain, which correspond to the feasible region of (P4).)

    For all feasible solutions $\bm{X}$ to (P4), we have $F(\bm{X};\rho)=F(\bm{X};0)$, so
    \[
    F(\bm{X}^*;0) = F(\bm{X}^*;\rho) \le F(\bm{X};\rho) = F(\bm{X};0),
    \]
    where the inequality is due to the fact that $\bm{X}^*$ is a feasible solution of (P4). This shows that $\bm{X}^*$ is an optimal solution of (P4). The proof is finished.

\section{Proof of Theorem \ref{thm:ConvMD}\label{app:thmConvMD}}

    Before proceeding with the proof, we introduce an equivalent representation of the stationary point defined in Definition \ref{defn:station}.

    \begin{lemma}[Equivalent Representation of the Stationary Point]
    Define $T_t^F(\bm{X}):=\argmin_{\bm{Y}\in \Delta_k^N}\{\left\langle t\bm{X}(2\bm{W}-\rho\bm{I}), \bm{Y}\right\rangle + \KL  (\bm{Y}, \bm{X})\}$. Then $ \KL  (T_t^F(\bm{X}^*),\bm{X}^*)=0$ if and only if $\bm{X}^*$ is a stationary point of problem P5 for all $t\in(0, \frac{1}{L+1})$.\label{lemma:station}
    \end{lemma}
    \begin{proof}
        For $\bm{X}^*$ such that $ \KL  (T_t^F(\bm{X}^*),\bm{X}^*)=0$, it also holds that 
        \begin{align*}
            \bm{X}^*=\argmin_{\bm{X}\in \Delta_k^N}\{\left\langle t\bm{X}^*(2\bm{W}-\rho\bm{I}), \bm{X}\right\rangle + \KL  (\bm{X}, \bm{X}^*)\}.
        \end{align*}
    
        As the inner function is convex with respect to $\bm{X}$ and $\Delta_k^N$ is convex, the optimal $\bm{X}^*$ is a stationary point of this problem. Therefore, $\forall \bm{X}\in \Delta_k^N$, 
        {\small
        \begin{align*}
            &\left\langle\frac{\partial(\left\langle t\bm{X}^*(2\bm{W}-\rho\bm{I}), \bm{X}\right\rangle + \KL  (\bm{X},\bm{X}^*))}{\partial \bm{X}}\vert_{\bm{X}=\bm{X}^*}, \bm{X}-\bm{X}^*\right\rangle \ge 0. 
        \end{align*}  }
        This implies
        \begin{align*}    
            \left\langle\bm{X}^*(2\bm{W}-\rho\bm{I}), \bm{X}-\bm{X}^*\right\rangle \ge 0,
        \end{align*}
        showing that $\bm{X}^*$ is a stationary point of problem P5.
    \end{proof}
    
    Then, we can introduce the following lemmas to analyze the benign properties of KL divergence and the objective function.

    \begin{lemma}[Benign Properties of KL Bregman Divergence \cite{beck2017first}]\label{lem:LemKL} The KL Bregman divergence (\ref{divergence}) can be expressed as
    \begin{equation}
        \begin{aligned}
            \KL(\bm{X},\bm{Y})
            &= \omega(\bm{X})-\omega(\bm{Y}) -\left\langle\nabla\omega(\bm{Y}),\bm{X}-\bm{Y}\right\rangle,
        \end{aligned}
    \end{equation}
    where
    \begin{align}
    \omega(\bm{X})=\sum_{i=1}^k\sum_{j=1}^N[\bm{X}]_{i,j}\log [\bm{X}]_{i,j},
    \label{for:omega}
    \end{align}
    and $\bm{X}, \bm{X}^{(m)}\in \Delta_k^N$.
    \begin{enumerate}
        \item The function $\omega(\bm{X})$ is $1$-strongly convex on $\Delta_{k}^N$. 
        \item The function $\KL(\bm{X},\bm{Y})$ satisfies that $\KL(\bm{X},\bm{Y})\ge\frac{1}{2}\Vert\bm{X}-\bm{Y}\Vert_F^2$.
    \end{enumerate}
    \end{lemma}
    \begin{lemma}[Benign Properties of  Objective Function] The objective function $F(\bm{X};\rho)$ in the penalized problem P5 satisfies following properties.
    \begin{enumerate}
              \item\label{lem:LemObj-1} ($L$-smooth): $F(\bm{X};\rho)$ is $L$-smooth over $\Delta_k^N$ with respect to $\KL(\cdot)$  in a sense that 
        \begin{align}
            F(\bm{Y};\rho)\le &F(\bm{X};\rho) + \left\langle\nabla F(\bm{X};\rho),\bm{Y}-\bm{X} \right\rangle \notag\\
            &+ L\cdot\KL(\bm{Y}, \bm{X}), \label{l-smooth}       
        \end{align}
        where $L=\Vert 2\bm{W}-\rho \bm{I}\Vert_F$.
        \item\label{lem:LemObj-2} (Bounded): $F(\bm{X};\rho)$ is bounded over $\Delta_k^N$, which is
        \begin{align}
        -\frac{\rho}{2}N\le F(\bm{X};\rho)\le \Vert\bm{W}\Vert_{1,1}.
        \end{align}
    \end{enumerate}
    \label{lem:LemObj}
    \end{lemma}
    \begin{proof}
        \textbf{Part \ref{lem:LemObj-1}.} Considering that $\frac{\partial F(\bm{X};\rho)}{\partial \bm{X}}=\bm{X}(2\bm{W}-\rho \bm{I})$, we have 
        \begin{align}
            \Vert\nabla F(\bm{X};\rho)-\nabla F(\bm{Y};\rho)\Vert_F &=\Vert(\bm{X}-\bm{Y})(2\bm{W}-\rho\bm{I})\Vert_F\notag\\
            &\le \Vert\bm{X}-\bm{Y}\Vert_F\Vert(2\bm{W}-\rho\bm{I})\Vert_F.\notag
        \end{align}
        
        Thus, $F(\bm{X};\rho)$ is $\Vert(2\bm{W}-\rho\bm{I})\Vert_F$-smooth with respect to $\Vert\cdot\Vert_F$, which means 
        \begin{align*}
         F(\bm{Y};\rho)&\le F(\bm{X};\rho)+\left\langle\nabla F(\bm{X};\rho),\bm{Y}-\bm{X}\right\rangle\\
         &+\frac{\Vert 2\bm{W}-\rho\bm{I}\Vert_F}{2}\Vert Y-X\Vert_F^2 \\
        & \overset{(a)}{\le} F(\bm{X};\rho)+\left\langle\nabla F(\bm{X};\rho),\bm{Y}-\bm{X}\right\rangle\\
        &+\Vert 2\bm{W}-\rho\bm{I}\Vert_F \KL  (\bm{Y},\bm{X}),
        \end{align*}
        where (a) is due to Lemma \ref{lem:LemKL}.
        
        Therefore, $F(\bm{X};\rho)$ is $\Vert 2\bm{W}-\rho \bm{I}\Vert_F$-smooth with respect to $ \KL  (\cdot)$.
        
        \textbf{Part \ref{lem:LemObj-2}.} For the upper bound, we derive
        \begin{align*}
            F(\bm{X};\rho)\le \text{Tr}(\bm{X}\bm{W}\bm{X}^T) 
            =\sum_{i=1}^n\sum_{j=1}^n [\bm{W}]_{i,j}x_i^Tx_j 
            \le \Vert\bm{W}\Vert_{1,1}.
        \end{align*}
        
        For the lower bound, we obtain
        \begin{align*}
            F(\bm{X};\rho)&\ge -\frac{\rho}{2}\text{Tr}(\bm{X}\bm{X}^T) 
            =-\frac{\rho}{2}\sum_{i=1}^N\Vert\bm{x}_i\Vert_2^2 \\
            &\ge -\frac{\rho}{2}\sum_{i=1}^N\Vert\bm{x}_i\Vert_1^2
            = -\frac{\rho}{2}N.
        \end{align*}
    \end{proof}

    Now, we establish a sufficient decrease lemma for MD. 

    \begin{lemma}[Sufficient Decrease Lemma]
    For the step size $\{t_m\}$ such that $0<t_m<\frac{1}{L+1}$, $\{\bm{X}^{(m)}\}_{m=0}^{T}$ is sequence generated by MD. Then, $F(\bm{X}^{(m+1)};\rho)\le F(\bm{X}^{(m)};\rho)-(\frac{1}{t_m}-L) \KL  (\bm{X}^{(m+1)},\bm{X}^{(m)})$.
    \label{lemma:2}
    \end{lemma}
    \begin{proof}
        We begin by proving a relevant inequality. Since $\omega$ is convex, we have
        \begin{align}
            \omega(\bm{X})\ge \omega(\bm{Y})+\left\langle\nabla\omega(\bm{Y}),\bm{X}-\bm{Y}\right\rangle, \quad \forall \bm{X},\bm{Y}\in\Delta_k^N.
        \end{align}
        This inequality leads to 
        \begin{align}
            \left\langle\nabla\omega(\bm{Y})-\nabla\omega(\bm{X}), \bm{Y}-\bm{X}\right\rangle\ge  \KL  (\bm{Y},\bm{X}),\forall \bm{X},\bm{Y}\in\Delta_k^N.\label{ineqn}
        \end{align}
        Considering that each iteration (\ref{for:prox:md}) is convex with respect to $\bm{X}$ and $\Delta_k^N$ is convex, the optimal $\bm{X}^{(m+1)}$ is a stationary point of the problem, implying that $\forall \bm{X}\in \Delta_k^N$, 
        \begin{align}
            &\left\langle 
                t_m\nabla F(\bm{X}^{(m)};\rho)
                + \nabla\omega(\bm{X}^{(m+1)})
                - \nabla\omega(\bm{X}^{(m)}), \right. \nonumber\\
            &\left. \quad \bm{X} - \bm{X}^{(m+1)} 
            \right\rangle \ge 0. \label{ineqn2}
        \end{align}
        Substituting $\bm{X}=\bm{X}^{(m)}$ into (\ref{ineqn2}) and applying (\ref{ineqn}), we obtain
        \begin{align}
            &\left\langle 
                \nabla F(\bm{X}^{(m)};\rho), 
                \bm{X}^{(m+1)} - \bm{X}^{(m)} 
            \right\rangle \nonumber\\
            &\quad \le -\frac{1}{t_m} \, \KL(\bm{X}^{(m+1)}, \bm{X}^{(m)}).
            \label{ineqn3}
        \end{align}
    
        Since $F(\bm{X};\rho)$ is $L$-smooth, combining (\ref{l-smooth}) and (\ref{ineqn3}), we have
        \begin{align}
            &F(\bm{X}^{(m+1)};\rho) \le F(\bm{X}^{(m)};\rho) \nonumber\\
            &\quad -\left(\frac{1}{t_m} - L\right) \KL(\bm{X}^{(m+1)}, \bm{X}^{(m)}).
            \label{for:lemma2}
        \end{align}
    \end{proof}

    Finally, by applying the sufficient decrease lemma and summing (\ref{for:lemma2}) from $0$ to $T-1$, we obtain:
    \begin{align}
        &\sum_{m=0}^{T-1} \left(\frac{1}{t_m} - L\right) 
        \KL(\bm{X}^{(m+1)}, \bm{X}^{(m)}) \nonumber\\
        &\quad \le F(\bm{X}^{(0)};\rho) - F(\bm{X}^{(T)};\rho).
    \end{align}
    which indicates
    \begin{align}
        &\min_{m=0,\cdots, T-1} \KL(\bm{X}^{(m+1)}, \bm{X}^{(m)})
        \sum_{m=0}^{T-1} \left(\frac{1}{t_m} - L\right) \nonumber\\
        &\quad \le F_{\max} - F_{\min}. 
    \end{align}
    By Lemma \ref{lem:LemObj}, supplementing $F_{\max}=\Vert\bm{W}\Vert_{1,1}$, $F_{\min}=-\frac{\rho N}{2}$, and choosing $t_m=\frac{1}{L+m+1}$, we derive
    \begin{align}
    &\min_{m=0, \cdots, T-1} \KL  (\bm{X}^{(m+1)}, \bm{X}^{(m)}) \\
    &\quad \le \frac{2\Vert\bm{W}\Vert_{1,1}+\rho N}{T(T+1)} \notag 
    \le \epsilon.\label{for:theo3}
    \end{align}

   Finally, we obtain the convergence rate as
    \begin{align}
    T\ge\sqrt{\frac{2\Vert\bm{W}\Vert_{1,1}+\rho N}{\epsilon}}\approx \mathcal{O}\left(\sqrt{\frac{1}{\epsilon}}\right).
    \end{align}
}


{
\bibliographystyle{IEEEtran}	
 \bibliography{IEEEabrv,Reference}
 }

\clearpage
\onecolumn

\begin{center}
    \LARGE \bfseries Supplementary Materials of ``Relaxation-Free Min-$k$-Partition\\ for PCI Assignment in 5G Networks''\\[2ex]
    \normalsize Yeqing Qiu, Chengpiao Huang, Ye Xue, Zhipeng Jiang, Dong Zhang, and Zhi-Quan Luo\\[1ex]
\end{center}

\vspace{2em}

\setcounter{section}{0}                 

\section{The full algorithm of the greedy quotient re-assignment for PCI range constraint satisfaction}

\begin{algorithm}[H]
\KwData{Set of neighboring and second-order neighboring relations $\mathcal{E}$, mod-$30$ clusters $\{\mathcal{C}_r\}_{r=0}^{29}$, sets of constraint-violating cells $\{\mathcal{I}_r\}_{r=0}^{29}$ within the mod-$30$ clusters, mod-$30$ value assignment $\bm{r} \in \mathbb{Z}^N_{30}$, current quotient assignment $\bm{q}$.}
\KwResult{Re-assigned quotient $\widetilde{\bm{q}}$.}

\For{$r=0$ \KwTo $29$}{
    \For{$i\in\mathcal{I}_r$}{$\deg_i \leftarrow \vert \{ j\in\mathcal{C}_r: (i,j)\in\mathcal{E}, \, q_i=q_j \} \vert$}
    Sort $\mathcal{I}_r$ in decreasing order of $\{\deg_v\}_{v\in\mathcal{I}_r}$ \\
    $\mathcal{J}_r \leftarrow \mathcal{C}_r\backslash\mathcal{I}_r$~~~ \textit{// set of constraint-satisfying cells} \\
    \For{$i\in\mathcal{J}_r$}{
    $\widetilde{q}_i\leftarrow q_i$
    }
    \For{$i\in\mathcal{I}_r$}{
    $\displaystyle \widetilde{q}_{i} \leftarrow
    \argmin_{q:\,30q+r_i\le 1007}
    \left\vert
    \left\{
    j\in\mathcal{J}_r: (i,j)\in\mathcal{E}, \, q=\widetilde{q}_j
    \right\}
    \right\vert$ \\[4pt]
    $\mathcal{J}_r\leftarrow\mathcal{J}_r\cup\{i\}$
    }
}
\KwRet{$\widetilde{\bm{q}}=\{ \widetilde{q}_i \}_{i\in\mathcal{V}}$}
\caption{Greedy Quotient Re-assignment for PCI Range Constraint Satisfaction\label{alg:repair}}
\end{algorithm}

\section{Ablation Study on Hyperparameters}

To assess the robustness of the hyperparameters, we conducted an ablation study on $\mathrm{RGG}(200, 0.3)$ graphs by scaling each hyperparameter up and down. As shown in Table~\ref{tab:sensitivity}, our methods remain stable under such perturbations.

\begin{table}[H]
\centering
\caption{Sensitivity analysis of GP-PMD on $(N=200, r=0.3)$ RGG instances. The original setting is $\epsilon_1=10^{-5}, \epsilon_2=10^{-10},\gamma=1.1, \rho_0=10^{-8}$. We perturb one parameter each time and test the overall performance. }
\label{tab:sensitivity}
\begin{tabular}{@{}lccccc@{}}
\toprule
\textbf{Setting} & \textbf{Collision $\downarrow$} & \textbf{Confusion $\downarrow$} & \textbf{Mod-3 $\downarrow$} & \textbf{Mod-30 $\downarrow$} & \textbf{Time (s) $\downarrow$} \\
\midrule
Original                & 0.0000 ± 0.0000 & 0.0000 ± 0.0000 & 31065.5022 ± 1052.4893 & 86.7920 ± 59.3875  & 195.6269 ± 15.6531 \\
$\epsilon_1 \times 10$  & 0.0000 ± 0.0000 & 0.0000 ± 0.0000 & 31059.7002 ± 1050.3168 & 117.8055 ± 57.7300 & 200.6586 ± 21.1356 \\
$\epsilon_1 \times 0.1$ & 0.0000 ± 0.0000 & 0.0000 ± 0.0000 & 31044.0332 ± 1046.6022 & 73.9605 ± 52.5553  & 200.6314 ± 14.4452 \\
$\epsilon_2 \times 10$  & 0.0000 ± 0.0000 & 0.0000 ± 0.0000 & 31065.5022 ± 1052.4893 & 86.7920 ± 59.3875  & 195.9044 ± 16.3598 \\
$\epsilon_3 \times 0.1$ & 0.0000 ± 0.0000 & 0.0000 ± 0.0000 & 31065.5022 ± 1052.4893 & 86.7920 ± 59.3875  & 195.7012 ± 15.9792 \\
$\gamma \times 0.918$   & 0.0000 ± 0.0000 & 0.0000 ± 0.0000 & 31065.5022 ± 1052.4893 & 84.6270 ± 57.6500  & 202.4566 ± 15.7155 \\
$\gamma \times 1.1$     & 0.0000 ± 0.0000 & 0.0000 ± 0.0000 & 31068.1502 ± 1053.6886 & 90.6242 ± 58.4496  & 194.0473 ± 16.5126 \\
$\rho_0 \times 10$      & 0.0000 ± 0.0000 & 0.0000 ± 0.0000 & 31065.5022 ± 1052.4893 & 86.7920 ± 59.3875  & 196.0779 ± 16.4285 \\
$\rho_0 \times 0.1$     & 0.0000 ± 0.0000 & 0.0000 ± 0.0000 & 31065.5022 ± 1052.4893 & 86.7920 ± 59.3875  & 197.1436 ± 16.4936 \\
\bottomrule
\end{tabular}
\end{table}

\section{Detailed Results on Synthetic Data}

To provide a more comprehensive view of the experimental results on synthetic data, we present detailed box plots in Figs.~\ref{fig:boxplot_coll}, \ref{fig:boxplot_conf}, \ref{fig:boxplot_mod3}, \ref{fig:boxplot_mod30}, and \ref{fig:boxplot_time}, corresponding respectively to collision, confusion, mod-3 interference, mod-30 interference, and computational time. Note that the absence of visible boxes in the plots indicates that the corresponding metric attains zero across all test cases, reflecting a consistently optimal outcome for that criterion.

\begin{figure}[H]
    \centering
    \includegraphics[width=0.75\linewidth]{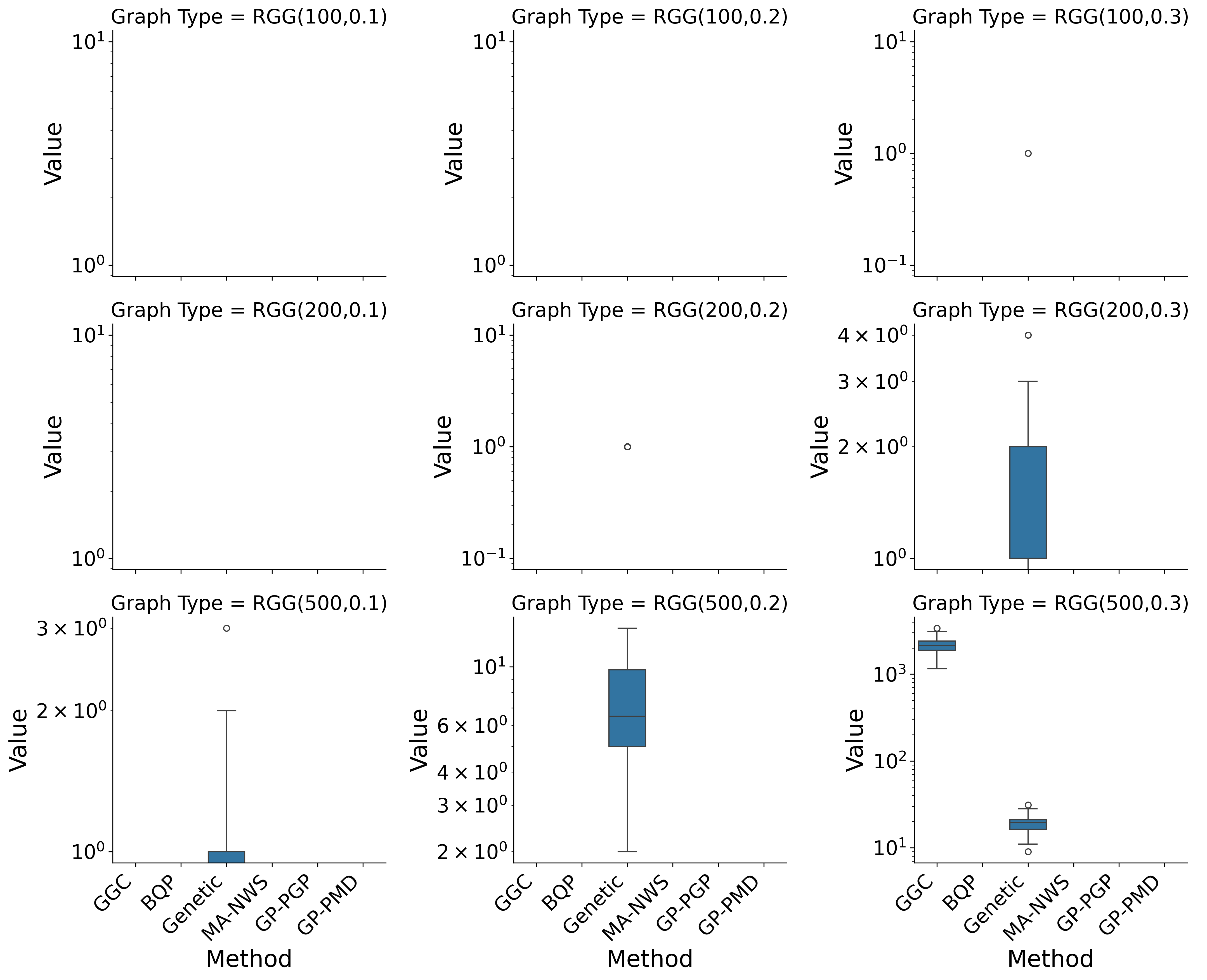}
    \caption{Boxplot of collision across methods and graph types.}
    \label{fig:boxplot_coll}
\end{figure}
\begin{figure}[H]
    \centering
    \includegraphics[width=0.75\linewidth]{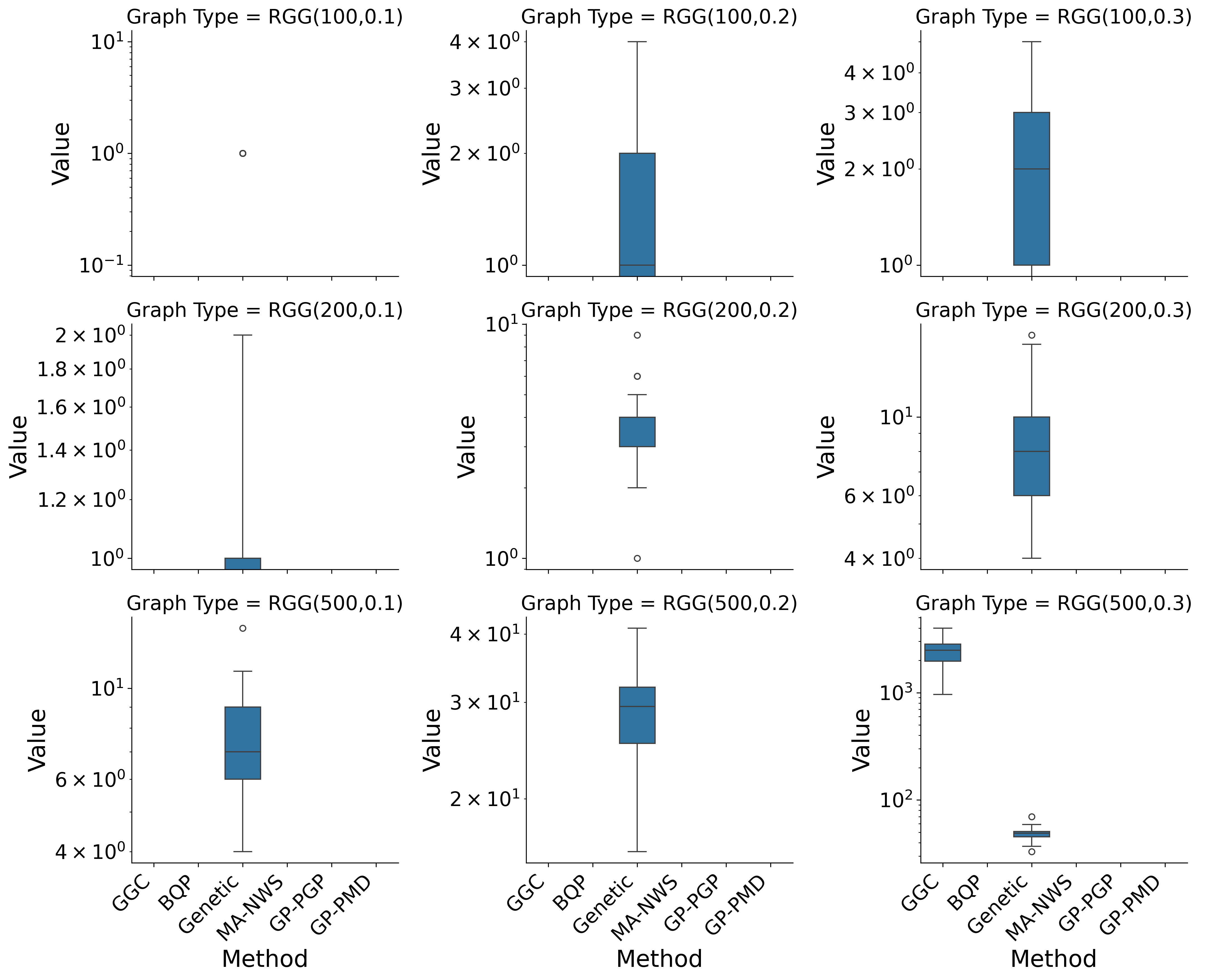}
    \caption{Boxplot of confusion across methods and graph types.}
    \label{fig:boxplot_conf}
\end{figure}
\begin{figure}[H]
    \centering
    \includegraphics[width=0.75\linewidth]{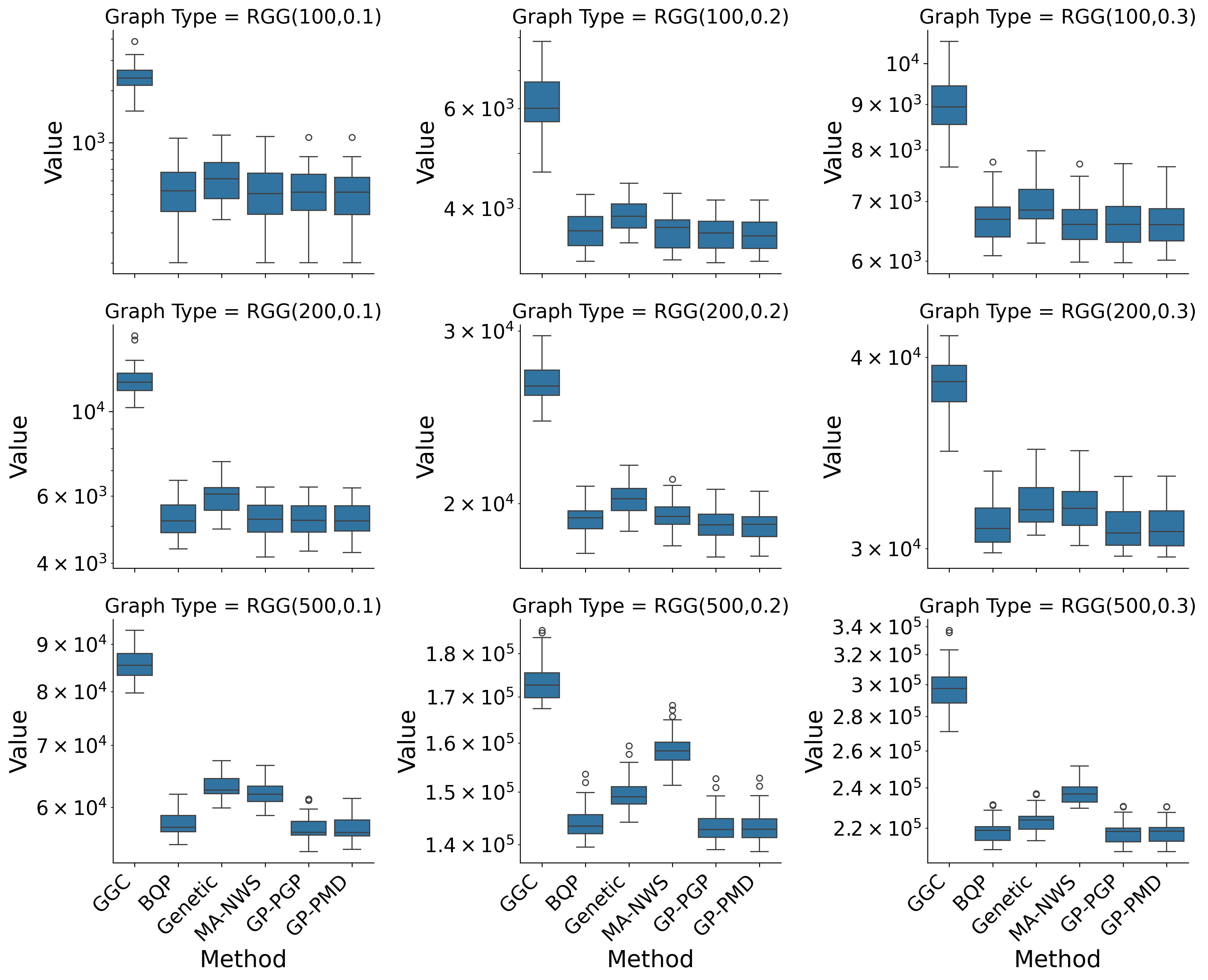}
    \caption{Boxplot of mod-3 interference across methods and graph types.}
    \label{fig:boxplot_mod3}
\end{figure}
\begin{figure}[H]
    \centering
    \includegraphics[width=0.75\linewidth]{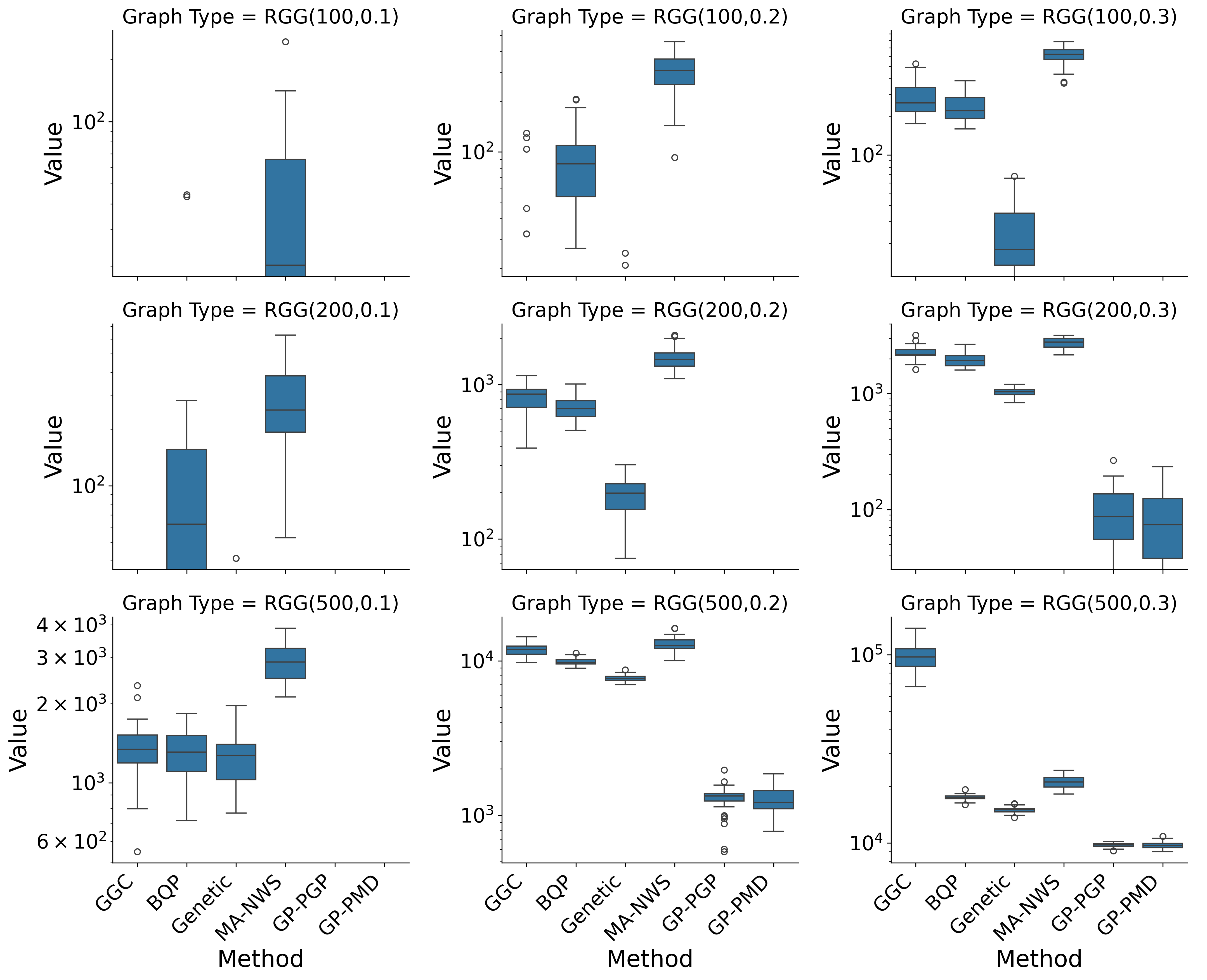}
    \caption{Boxplot of mod-30 interference across methods and graph types.}
    \label{fig:boxplot_mod30}
\end{figure}
\begin{figure}[H]
    \centering
    \includegraphics[width=0.75\linewidth]{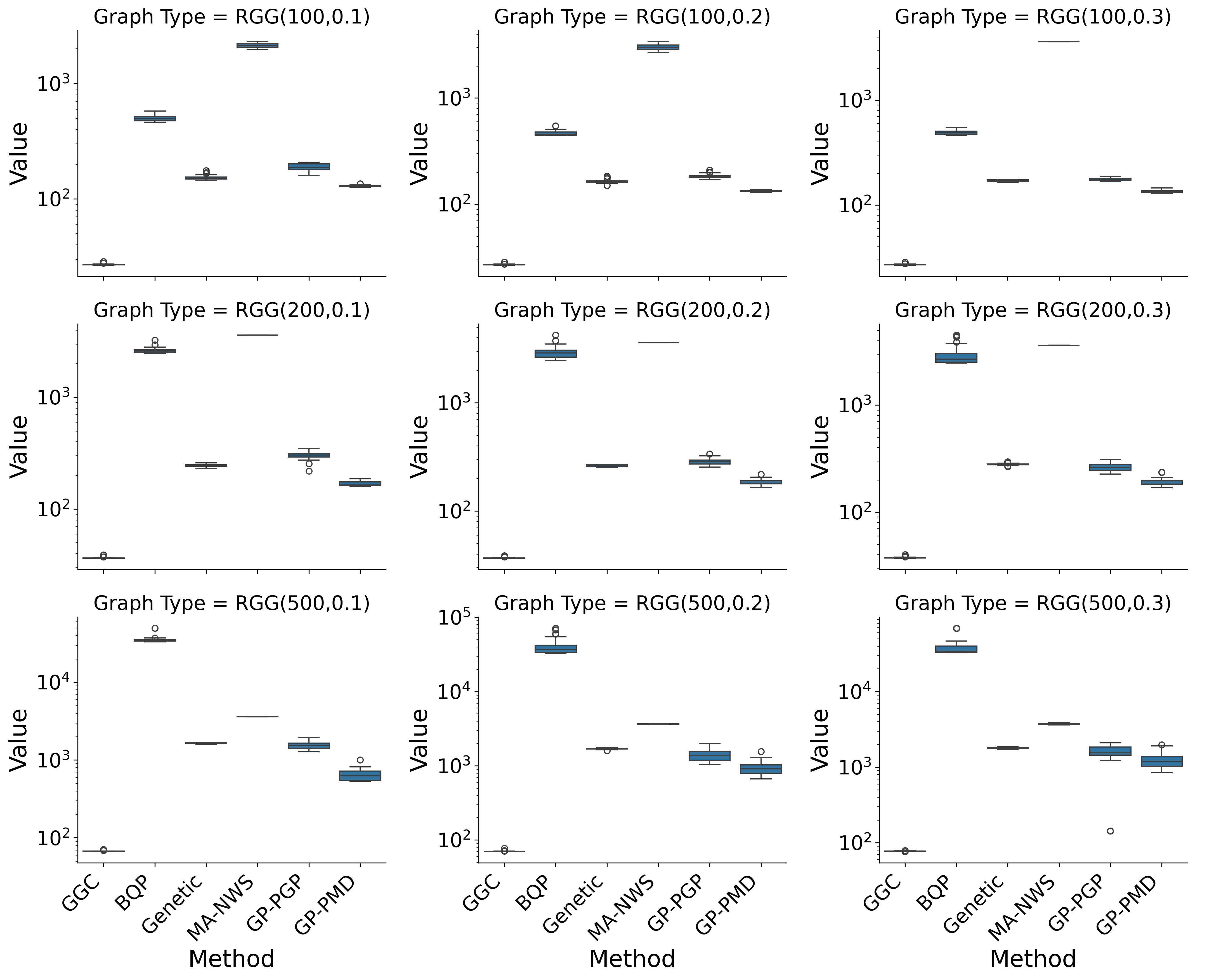}
    \caption{Boxplot of computational time across methods and graph types.}
    \label{fig:boxplot_time}
\end{figure}

\section{Detailed Results of statistical hypothesis tests}

To assess statistical significance, we applied the non-parametric Mann-Whitney U test. A $p$-value less than $0.05$ indicates significant performance differences between algorithms at the $0.05$ significance level. Moreover, to quantify the strength of these differences, we employed the Vargha-Delaney A test, which estimates the probability that one method outperforms another. An $A$-value above $0.5$ favors our method (GP-PMD); $A = 1$ implies complete dominance and $A = 0.5$ indicates no significant difference. Tables in Figs. ~\ref{tab:testing_coll}, \ref{tab:testing_conf}, \ref{tab:testing_mod3}, \ref{tab:testing_mod30}, \ref{tab:testing_time}  report $p$-value and A statistics for each algorithm pair across all graph types and objectives. Results confirm that GP-PMD statistically and practically outperforms all baseline methods.

\begin{figure}[H]
\centering
\renewcommand{\arraystretch}{1.1}
\begin{tabular}{ccc}
\begin{minipage}[t]{0.3\textwidth}
\centering
\fontsize{7pt}{9pt}\selectfont
\textbf{RGG(100, 0.1)} \\
\begin{tabular}{lccc}
\toprule
Method & $p$-value & $A$ & Winner \\
\midrule
GGC     & $1.00$ & $0.50$ & Equal \\
Genetic & $1.00$ & $0.50$ & Equal \\
BQP     & $1.00$ & $0.50$ & Equal \\
MA-NWS  & $1.00$ & $0.50$ & Equal \\
GP-PGP  & $1.00$ & $0.50$ & Equal \\
\bottomrule
\end{tabular}
\end{minipage}
&
\begin{minipage}[t]{0.3\textwidth}
\centering
\fontsize{7pt}{9pt}\selectfont
\textbf{RGG(100, 0.2)} \\
\begin{tabular}{lccc}
\toprule
Method & $p$-value & $A$ & Winner \\
\midrule
GGC     & $1.00$ & $0.50$ & Equal \\
Genetic & $1.00$ & $0.50$ & Equal \\
BQP     & $1.00$ & $0.50$ & Equal \\
MA-NWS  & $1.00$ & $0.50$ & Equal \\
GP-PGP  & $1.00$ & $0.50$ & Equal \\
\bottomrule
\end{tabular}
\end{minipage}
&
\begin{minipage}[t]{0.3\textwidth}
\centering
\fontsize{7pt}{9pt}\selectfont
\textbf{RGG(100, 0.3)} \\
\begin{tabular}{lccc}
\toprule
Method & $p$-value & $A$ & Winner \\
\midrule
GGC     & $1.00$ & $0.50$ & Equal \\
Genetic & $0.33$ & $0.52$ & GP-PMD \\
BQP     & $1.00$ & $0.50$ & Equal \\
MA-NWS  & $1.00$ & $0.50$ & Equal \\
GP-PGP  & $1.00$ & $0.50$ & Equal \\
\bottomrule
\end{tabular}
\end{minipage}
\\[1em]

\begin{minipage}[t]{0.3\textwidth}
\centering
\fontsize{7pt}{9pt}\selectfont
\textbf{RGG(200, 0.1)} \\
\begin{tabular}{lccc}
\toprule
Method & $p$-value & $A$ & Winner \\
\midrule
GGC     & $1.00$ & $0.50$ & Equal \\
Genetic & $1.00$ & $0.50$ & Equal \\
BQP     & $1.00$ & $0.50$ & Equal \\
MA-NWS  & $1.00$ & $0.50$ & Equal \\
GP-PGP  & $1.00$ & $0.50$ & Equal \\
\bottomrule
\end{tabular}
\end{minipage}
&
\begin{minipage}[t]{0.3\textwidth}
\centering
\fontsize{7pt}{9pt}\selectfont
\textbf{RGG(200, 0.2)} \\
\begin{tabular}{lccc}
\toprule
Method & $p$-value & $A$ & Winner \\
\midrule
GGC     & $1.00$ & $0.50$ & Equal \\
Genetic & $0.01$ & $0.60$ & GP-PMD \\
BQP     & $1.00$ & $0.50$ & Equal \\
MA-NWS  & $1.00$ & $0.50$ & Equal \\
GP-PGP  & $1.00$ & $0.50$ & Equal \\
\bottomrule
\end{tabular}
\end{minipage}
&
\begin{minipage}[t]{0.3\textwidth}
\centering
\fontsize{7pt}{9pt}\selectfont
\textbf{RGG(200, 0.3)} \\
\begin{tabular}{lccc}
\toprule
Method & $p$-value & $A$ & Winner \\
\midrule
GGC     & $1.00$ & $0.50$ & Equal \\
Genetic & $\ll 0.01$ & $0.92$ & GP-PMD \\
BQP     & $1.00$ & $0.50$ & Equal \\
MA-NWS  & $1.00$ & $0.50$ & Equal \\
GP-PGP  & $1.00$ & $0.50$ & Equal \\
\bottomrule
\end{tabular}
\end{minipage}
\\[1em]

\begin{minipage}[t]{0.3\textwidth}
\centering
\fontsize{7pt}{9pt}\selectfont
\textbf{RGG(500, 0.1)} \\
\begin{tabular}{lccc}
\toprule
Method & $p$-value & $A$ & Winner \\
\midrule
GGC     & $1.00$ & $0.50$ & Equal \\
Genetic & $\ll 0.01$ & $0.72$ & GP-PMD \\
BQP     & $1.00$ & $0.50$ & Equal \\
MA-NWS  & $1.00$ & $0.50$ & Equal \\
GP-PGP  & $1.00$ & $0.50$ & Equal \\
\bottomrule
\end{tabular}
\end{minipage}
&
\begin{minipage}[t]{0.3\textwidth}
\centering
\fontsize{7pt}{9pt}\selectfont
\textbf{RGG(500, 0.2)} \\
\begin{tabular}{lccc}
\toprule
Method & $p$-value & $A$ & Winner \\
\midrule
GGC     & $1.00$ & $0.50$ & Equal \\
Genetic & $\ll 0.01$ & $1.00$ & GP-PMD \\
BQP     & $1.00$ & $0.50$ & Equal \\
MA-NWS  & $1.00$ & $0.50$ & Equal \\
GP-PGP  & $1.00$ & $0.50$ & Equal \\
\bottomrule
\end{tabular}
\end{minipage}
&
\begin{minipage}[t]{0.3\textwidth}
\centering
\fontsize{7pt}{9pt}\selectfont
\textbf{RGG(500, 0.3)} \\
\begin{tabular}{lccc}
\toprule
Method & $p$-value & $A$ & Winner \\
\midrule
GGC     & $\ll 0.01$ & $1.00$ & GP-PMD \\
Genetic & $\ll 0.01$ & $1.00$ & GP-PMD \\
BQP     & $1.00$ & $0.50$ & Equal \\
MA-NWS  & $1.00$ & $0.50$ & Equal \\
GP-PGP  & $1.00$ & $0.50$ & Equal \\
\bottomrule
\end{tabular}
\end{minipage}
\end{tabular}
\caption{GP-PMD versus other algorithms on number of collisions}
\label{tab:testing_coll}
\end{figure}

\begin{figure}[H]
\centering
\renewcommand{\arraystretch}{1.1}
\begin{tabular}{ccc}
\begin{minipage}[t]{0.3\textwidth}
\centering
\fontsize{7pt}{9pt}\selectfont
\textbf{RGG(100, 0.1)} \\
\begin{tabular}{lccc}
\toprule
Method & $p$-value & $A$ & Winner \\
\midrule
GGC     & $1.00$ & $0.50$ & Equal \\
Genetic & $0.01$ & $0.60$ & GP-PMD \\
BQP     & $1.00$ & $0.50$ & Equal \\
MA-NWS  & $1.00$ & $0.50$ & Equal \\
GP-PGP  & $1.00$ & $0.50$ & Equal \\
\bottomrule
\end{tabular}
\end{minipage}
&
\begin{minipage}[t]{0.3\textwidth}
\centering
\fontsize{7pt}{9pt}\selectfont
\textbf{RGG(100, 0.2)} \\
\begin{tabular}{lccc}
\toprule
Method & $p$-value & $A$ & Winner \\
\midrule
GGC     & $1.00$ & $0.50$ & Equal \\
Genetic & $\ll 0.01$ & $0.85$ & GP-PMD \\
BQP     & $1.00$ & $0.50$ & Equal \\
MA-NWS  & $1.00$ & $0.50$ & Equal \\
GP-PGP  & $1.00$ & $0.50$ & Equal \\
\bottomrule
\end{tabular}
\end{minipage}
&
\begin{minipage}[t]{0.3\textwidth}
\centering
\fontsize{7pt}{9pt}\selectfont
\textbf{RGG(100, 0.3)} \\
\begin{tabular}{lccc}
\toprule
Method & $p$-value & A & Winner \\
\midrule
GGC     & $1.00$ & $0.50$ & Equal \\
Genetic & $\ll 0.01$ & $0.93$ & GP-PMD \\
BQP     & $1.00$ & $0.50$ & Equal \\
MA-NWS  & $1.00$ & $0.50$ & Equal \\
GP-PGP  & $1.00$ & $0.50$ & Equal \\
\bottomrule
\end{tabular}
\end{minipage}
\\[1em]

\begin{minipage}[t]{0.3\textwidth}
\centering
\fontsize{7pt}{9pt}\selectfont
\textbf{RGG(200, 0.1)} \\
\begin{tabular}{lccc}
\toprule
Method & $p$-value & $A$ & Winner \\
\midrule
GGC     & $1.00$ & $0.50$ & Equal \\
Genetic & $1.00$ & $0.50$ & Equal \\
BQP     & $\ll 0.01$ & $0.72$ & GP-PMD \\
MA-NWS  & $1.00$ & $0.50$ & Equal \\
GP-PGP  & $1.00$ & $0.50$ & Equal \\
\bottomrule
\end{tabular}
\end{minipage}
&
\begin{minipage}[t]{0.3\textwidth}
\centering
\fontsize{7pt}{9pt}\selectfont
\textbf{RGG(200, 0.2)} \\
\begin{tabular}{lccc}
\toprule
Method & $p$-value & $A$ & Winner \\
\midrule
GGC     & $1.00$ & $0.50$ & Equal \\
Genetic & $\ll0.01$ & $0.98$ & GP-PMD \\
BQP     & $1.00$ & $0.50$ & Equal \\
MA-NWS  & $1.00$ & $0.50$ & Equal \\
GP-PGP  & $1.00$ & $0.50$ & Equal \\
\bottomrule
\end{tabular}
\end{minipage}
&
\begin{minipage}[t]{0.3\textwidth}
\centering
\fontsize{7pt}{9pt}\selectfont
\textbf{RGG(200, 0.3)} \\
\begin{tabular}{lccc}
\toprule
Method & $p$-value & $A$ & Winner \\
\midrule
GGC     & $1.00$ & $0.50$ & Equal \\
Genetic & $\ll 0.01$ & $1.00$ & GP-PMD \\
BQP     & $1.00$ & $0.50$ & Equal \\
MA-NWS  & $1.00$ & $0.50$ & Equal \\
GP-PGP  & $1.00$ & $0.50$ & Equal \\
\bottomrule
\end{tabular}
\end{minipage}
\\[1em]

\begin{minipage}[t]{0.3\textwidth}
\centering
\fontsize{7pt}{9pt}\selectfont
\textbf{RGG(500, 0.1)} \\
\begin{tabular}{lccc}
\toprule
Method & $p$-value & $A$ & Winner \\
\midrule
GGC     & $1.00$ & $0.50$ & Equal \\
Genetic & $\ll 0.01$ & $1.00$ & GP-PMD \\
BQP     & $1.00$ & $0.50$ & Equal \\
MA-NWS  & $1.00$ & $0.50$ & Equal \\
GP-PGP  & $1.00$ & $0.50$ & Equal \\
\bottomrule
\end{tabular}
\end{minipage}
&
\begin{minipage}[t]{0.3\textwidth}
\centering
\fontsize{7pt}{9pt}\selectfont
\textbf{RGG(500, 0.2)} \\
\begin{tabular}{lccc}
\toprule
Method & $p$-value & A & Winner \\
\midrule
GGC     & $1.00$ & $0.50$ & Equal \\
Genetic & $\ll 0.01$ & $1.00$ & GP-PMD \\
BQP     & $1.00$ & $0.50$ & Equal \\
MA-NWS  & $1.00$ & $0.50$ & Equal \\
GP-PGP  & $1.00$ & $0.50$ & Equal \\
\bottomrule
\end{tabular}
\end{minipage}
&
\begin{minipage}[t]{0.3\textwidth}
\centering
\fontsize{7pt}{9pt}\selectfont
\textbf{RGG(500, 0.3)} \\
\begin{tabular}{lccc}
\toprule
Method & $p$-value & $A$ & Winner \\
\midrule
GGC     & $\ll 0.01$ & $1.00$ & GP-PMD \\
Genetic & $\ll 0.01$ & $1.00$ & GP-PMD \\
BQP     & $1.00$ & $0.50$ & Equal \\
MA-NWS  & $1.00$ & $0.50$ & Equal \\
GP-PGP  & $1.00$ & $0.50$ & Equal \\
\bottomrule
\end{tabular}
\end{minipage}
\end{tabular}
\caption{GP-PMD versus other algorithms on number of confusions}
\label{tab:testing_conf}
\end{figure}

\begin{figure}[H]
\centering
\renewcommand{\arraystretch}{1.1}
\begin{tabular}{ccc}
\begin{minipage}[t]{0.3\textwidth}
\centering
\fontsize{7pt}{9pt}\selectfont
\textbf{RGG(100, 0.1)} \\
\begin{tabular}{lccc}
\toprule
Method & $p$-value & $A$ & Winner \\
\midrule
GGC     & $\ll 0.01$ & $1.00$ & GP-PMD \\
Genetic & $0.07$ & $0.64$ & GP-PMD \\
BQP     & $0.77$ & $0.52$ & GP-PMD \\
MA-NWS  & $0.90$ & $0.51$ & GP-PMD \\
GP-PGP  & $0.92$ & $0.51$ & GP-PMD \\
\bottomrule
\end{tabular}
\end{minipage}
&
\begin{minipage}[t]{0.3\textwidth}
\centering
\fontsize{7pt}{9pt}\selectfont
\textbf{RGG(100, 0.2)} \\
\begin{tabular}{lccc}
\toprule
Method & $p$-value & $A$ & Winner \\
\midrule
GGC     & $\ll 0.01$ & $1.00$ & GP-PMD \\
Genetic & $\ll 0.01$ & $0.74$ & GP-PMD \\
BQP     & $0.38$ & $0.57$ & GP-PMD \\
MA-NWS  & $0.46$ & $0.56$ & GP-PMD \\
GP-PGP  & $1.00$ & $0.50$ & Equal \\
\bottomrule
\end{tabular}
\end{minipage}
&
\begin{minipage}[t]{0.3\textwidth}
\centering
\fontsize{7pt}{9pt}\selectfont
\textbf{RGG(100, 0.3)} \\
\begin{tabular}{lccc}
\toprule
Method & $p$-value & $A$ & Winner \\
\midrule
GGC     & $\ll 0.01$ & $1.00$ & GP-PMD \\
Genetic & $0.01$ & $0.71$ & GP-PMD \\
BQP     & $0.49$ & $0.55$ & GP-PMD \\
MA-NWS  & $0.96$ & $0.50$ & Equal \\
GP-PGP  & $0.98$ & $0.50$ & Equal \\
\bottomrule
\end{tabular}
\end{minipage}
\\[1em]

\begin{minipage}[t]{0.3\textwidth}
\centering
\fontsize{7pt}{9pt}\selectfont
\textbf{RGG(200, 0.1)} \\
\begin{tabular}{lccc}
\toprule
Method & $p$-value & $A$ & Winner \\
\midrule
GGC     & $\ll 0.01$ & $1.00$ & GP-PMD \\
Genetic & $\ll 0.01$ & $1.00$ & GP-PMD \\
BQP     & $0.86$ & $0.51$ & GP-PMD \\
MA-NWS  & $0.95$ & $0.51$ & GP-PMD \\
GP-PGP  & $0.92$ & $0.51$ & GP-PMD \\
\bottomrule
\end{tabular}
\end{minipage}
&
\begin{minipage}[t]{0.3\textwidth}
\centering
\fontsize{7pt}{9pt}\selectfont
\textbf{RGG(200, 0.2)} \\
\begin{tabular}{lccc}
\toprule
Method & $p$-value & $A$ & Winner \\
\midrule
GGC     & $\ll 0.01$ & $1.00$ & GP-PMD \\
Genetic & $\ll0.01$ & $0.86$ & GP-PMD \\
BQP     & $0.23$ & $0.59$ & GP-PMD \\
MA-NWS  & $0.04$ & $0.66$ & GP-PMD \\
GP-PGP  & $0.91$ & $0.51$ & GP-PMD \\
\bottomrule
\end{tabular}
\end{minipage}
&
\begin{minipage}[t]{0.3\textwidth}
\centering
\fontsize{7pt}{9pt}\selectfont
\textbf{RGG(200, 0.3)} \\
\begin{tabular}{lccc}
\toprule
Method & $p$-value & $A$ & Winner \\
\midrule
GGC     & $\ll 0.01$ & $1.00$ & GP-PMD \\
Genetic & $\ll 0.01$ & $0.76$ & GP-PMD \\
BQP     & $0.41$ & $0.56$ & GP-PMD \\
MA-NWS  & $\ll 0.01$ & $0.73$ & GP-PMD \\
GP-PGP  & $0.98$ & $0.50$ & Equal \\
\bottomrule
\end{tabular}
\end{minipage}
\\[1em]

\begin{minipage}[t]{0.3\textwidth}
\centering
\fontsize{7pt}{9pt}\selectfont
\textbf{RGG(500, 0.1)} \\
\begin{tabular}{lccc}
\toprule
Method & $p$-value & $A$ & Winner \\
\midrule
GGC     & $\ll 0.01$ & $1.00$ & GP-PMD \\
Genetic & $\ll 0.01$ & $0.99$ & GP-PMD \\
BQP     & $0.08$ & $0.63$ & GP-PMD \\
MA-NWS  & $\ll 0.01$ & $0.97$ & GP-PMD \\
GP-PGP  & $0.98$ & $0.50$ & Equal \\
\bottomrule
\end{tabular}
\end{minipage}
&
\begin{minipage}[t]{0.3\textwidth}
\centering
\fontsize{7pt}{9pt}\selectfont
\textbf{RGG(500, 0.2)} \\
\begin{tabular}{lccc}
\toprule
Method & $p$-value & $A$ & Winner \\
\midrule
GGC     & $\ll 0.01$ & $1.00$ & GP-PMD \\
Genetic & $\ll 0.01$ & $0.90$ & GP-PMD \\
BQP     & $0.29$ & $0.58$ & GP-PMD \\
MA-NWS  & $\ll 0.01$ & $1.00$ & GP-PMD \\
GP-PGP  & $0.99$ & $0.50$ & Equal \\
\bottomrule
\end{tabular}
\end{minipage}
&
\begin{minipage}[t]{0.3\textwidth}
\centering
\fontsize{7pt}{9pt}\selectfont
\textbf{RGG(500, 0.3)} \\
\begin{tabular}{lccc}
\toprule
Method & $p$-value & $A$ & Winner \\
\midrule
GGC     & $\ll 0.01$ & $1.00$ & GP-PMD \\
Genetic & $\ll 0.01$ & $0.78$ & GP-PMD \\
BQP     & $0.46$ & $0.56$ & GP-PMD \\
MA-NWS  & $\ll 0.01$ & $0.99$ & GP-PMD \\
GP-PGP  & $0.99$ & $0.50$ & Equal \\
\bottomrule
\end{tabular}
\end{minipage}
\end{tabular}
\caption{GP-PMD versus other algorithms on mod-3 interference}
\label{tab:testing_mod3}
\end{figure}

\begin{figure}[H]
\centering
\renewcommand{\arraystretch}{1.1}
\begin{tabular}{ccc}
\begin{minipage}[t]{0.3\textwidth}
\centering
\fontsize{7pt}{9pt}\selectfont
\textbf{RGG(100, 0.1)} \\
\begin{tabular}{lccc}
\toprule
Method & $p$-value & $A$ & Winner \\
\midrule
GGC     & $1.00$ & $0.50$ & Equal \\
Genetic & $1.00$ & $0.50$ & Equal \\
BQP     & $0.16$ & $0.53$ & GP-PMD \\
MA-NWS  & $\ll 0.01$ & $0.75$ & GP-PMD \\
GP-PGP  & $1.00$ & $0.50$ & Equal \\
\bottomrule
\end{tabular}
\end{minipage}
&
\begin{minipage}[t]{0.3\textwidth}
\centering
\fontsize{7pt}{9pt}\selectfont
\textbf{RGG(100, 0.2)} \\
\begin{tabular}{lccc}
\toprule
Method & $p$-value & $A$ & Winner \\
\midrule
GGC     & $0.02$ & $0.58$ & GP-PMD \\
Genetic & $0.16$ & $0.58$ & GP-PMD \\
BQP     & $\ll 0.01$ & $1.00$ & GP-PMD \\
MA-NWS  & $\ll 0.01$ & $1.00$ & GP-PMD \\
GP-PGP  & $1.00$ & $0.50$ & Equal \\
\bottomrule
\end{tabular}
\end{minipage}
&
\begin{minipage}[t]{0.3\textwidth}
\centering
\fontsize{7pt}{9pt}\selectfont
\textbf{RGG(100, 0.3)} \\
\begin{tabular}{lccc}
\toprule
Method & $p$-value & $A$ & Winner \\
\midrule
GGC     & $\ll 0.01$ & $1.00$ & GP-PMD \\
Genetic & $\ll 0.01$ & $0.92$& GP-PMD \\
BQP     & $\ll 0.01$ & $1.00$ & GP-PMD \\
MA-NWS  & $\ll 0.01$ & $1.00$ & GP-PMD \\
GP-PGP  & $1.00$ & $0.50$ & Equal \\
\bottomrule
\end{tabular}
\end{minipage}
\\[1em]

\begin{minipage}[t]{0.3\textwidth}
\centering
\fontsize{7pt}{9pt}\selectfont
\textbf{RGG(200, 0.1)} \\
\begin{tabular}{lccc}
\toprule
Method & $p$-value & $A$ & Winner \\
\midrule
GGC     & $1.00$ & $0.50$ & Equal \\
Genetic & $0.33$ & $0.52$ & GP-PMD \\
BQP     & $\ll 0.01$ & $0.83$ & GP-PMD \\
MA-NWS  & $\ll 0.01$ & $1.00$ & GP-PMD \\
GP-PGP  & $1.00$ & $0.50$ & Equal \\
\bottomrule
\end{tabular}
\end{minipage}
&
\begin{minipage}[t]{0.3\textwidth}
\centering
\fontsize{7pt}{9pt}\selectfont
\textbf{RGG(200, 0.2)} \\
\begin{tabular}{lccc}
\toprule
Method & $p$-value & $A$ & Winner \\
\midrule
GGC     & $\ll 0.01$ & $1.00$ & GP-PMD \\
Genetic & $\ll0.01$ & $1.00$ & GP-PMD \\
BQP     & $\ll0.01$ & $1.00$ & GP-PMD \\
MA-NWS  & $\ll0.01$ & $1.00$ & GP-PMD \\
GP-PGP  & $1.00$ & $0.50$ & Equal \\
\bottomrule
\end{tabular}
\end{minipage}
&
\begin{minipage}[t]{0.3\textwidth}
\centering
\fontsize{7pt}{9pt}\selectfont
\textbf{RGG(200, 0.3)} \\
\begin{tabular}{lccc}
\toprule
Method & $p$-value & $A$ & Winner \\
\midrule
GGC     & $\ll 0.01$ & $1.00$ & GP-PMD \\
Genetic & $\ll 0.01$ & $1.00$ & GP-PMD \\
BQP     & $\ll 0.01$ & $1.00$ & GP-PMD \\
MA-NWS  & $\ll 0.01$ & $1.00$ & GP-PMD \\
GP-PGP  & $0.43$ & $0.56$ & GP-PMD \\
\bottomrule
\end{tabular}
\end{minipage}
\\[1em]

\begin{minipage}[t]{0.3\textwidth}
\centering
\fontsize{7pt}{9pt}\selectfont
\textbf{RGG(500, 0.1)} \\
\begin{tabular}{lccc}
\toprule
Method & $p$-value & $A$ & Winner \\
\midrule
GGC     & $\ll 0.01$ & $1.00$ & GP-PMD \\
Genetic & $\ll 0.01$ & $1.00$ & GP-PMD \\
BQP     & $\ll 0.01$ & $1.00$ & GP-PMD \\
MA-NWS  & $\ll 0.01$ & $1.00$ & GP-PMD \\
GP-PGP  & $1.00$ & $0.50$ & Equal \\
\bottomrule
\end{tabular}
\end{minipage}
&
\begin{minipage}[t]{0.3\textwidth}
\centering
\fontsize{7pt}{9pt}\selectfont
\textbf{RGG(500, 0.2)} \\
\begin{tabular}{lccc}
\toprule
Method & $p$-value & $A$ & Winner \\
\midrule
GGC     & $\ll 0.01$ & $1.00$ & GP-PMD \\
Genetic & $\ll 0.01$ & $1.00$ & GP-PMD \\
BQP     & $\ll 0.01$ & $1.00$ & GP-PMD \\
MA-NWS  & $\ll 0.01$ & $1.00$ & GP-PMD \\
GP-PGP  & $0.31$ & $0.58$ & GP-PMD \\
\bottomrule
\end{tabular}
\end{minipage}
&
\begin{minipage}[t]{0.3\textwidth}
\centering
\fontsize{7pt}{9pt}\selectfont
\textbf{RGG(500, 0.3)} \\
\begin{tabular}{lccc}
\toprule
Method & $p$-value & A & Winner \\
\midrule
GGC     & $\ll 0.01$ & $1.00$ & GP-PMD \\
Genetic & $\ll 0.01$ & $1.00$ & GP-PMD \\
BQP     & $\ll 0.01$ & $1.00$ & GP-PMD \\
MA-NWS  & $\ll 0.01$ & $1.00$ & GP-PMD \\
GP-PGP  & $0.78$ & $0.52$ & GP-PMD \\
\bottomrule
\end{tabular}
\end{minipage}
\end{tabular}
\caption{GP-PMD versus other algorithms on mod-30 interference}\label{tab:testing_mod30}
\end{figure}

\begin{figure}[H]
\centering
\renewcommand{\arraystretch}{1.1}
\begin{tabular}{ccc}
\begin{minipage}[t]{0.3\textwidth}
\centering
\fontsize{7pt}{9pt}\selectfont
\textbf{RGG(100, 0.1)} \\
\begin{tabular}{lccc}
\toprule
Method & $p$-value & $A$ & Winner \\
\midrule
GGC     & $\ll 0.01$ & $0.00$ & GGC \\
Genetic & $\ll 0.01$ & $1.00$ & GP-PMD \\
BQP     & $\ll 0.01$ & $1.00$ & GP-PMD \\
MA-NWS  & $\ll 0.01$ & $1.00$ & GP-PMD \\
GP-PGP  & $\ll 0.01$ & $1.00$ & GP-PMD \\
\bottomrule
\end{tabular}
\end{minipage}
&
\begin{minipage}[t]{0.3\textwidth}
\centering
\fontsize{7pt}{9pt}\selectfont
\textbf{RGG(100, 0.2)} \\
\begin{tabular}{lccc}
\toprule
Method & $p$-value & $A$ & Winner \\
\midrule
GGC     & $\ll 0.01$ & $0.00$ & GGC \\
Genetic & $\ll 0.01$ & $1.00$ & GP-PMD \\
BQP     & $\ll 0.01$ & $1.00$ & GP-PMD \\
MA-NWS  & $\ll 0.01$ & $1.00$ & GP-PMD \\
GP-PGP  & $\ll 0.01$ & $1.00$ & GP-PMD \\
\bottomrule
\end{tabular}
\end{minipage}
&
\begin{minipage}[t]{0.3\textwidth}
\centering
\fontsize{7pt}{9pt}\selectfont
\textbf{RGG(100, 0.3)} \\
\begin{tabular}{lccc}
\toprule
Method & $p$-value & $A$ & Winner \\
\midrule
GGC     & $\ll 0.01$ & $0.00$ & GGC \\
Genetic & $\ll 0.01$ & $1.00$ & GP-PMD \\
BQP     & $\ll 0.01$ & $1.00$ & GP-PMD \\
MA-NWS  & $\ll 0.01$ & $1.00$ & GP-PMD \\
GP-PGP  & $\ll 0.01$ & $1.00$ & GP-PMD \\
\bottomrule
\end{tabular}
\end{minipage}
\\[1em]

\begin{minipage}[t]{0.3\textwidth}
\centering
\fontsize{7pt}{9pt}\selectfont
\textbf{RGG(200, 0.1)} \\
\begin{tabular}{lccc}
\toprule
Method & $p$-value & $A$ & Winner \\
\midrule
GGC     & $\ll 0.01$ & $0.00$ & GGC \\
Genetic & $\ll 0.01$ & $1.00$ & GP-PMD \\
BQP     & $\ll 0.01$ & $1.00$ & GP-PMD \\
MA-NWS  & $\ll 0.01$ & $1.00$ & GP-PMD \\
GP-PGP  & $\ll 0.01$ & $1.00$ & GP-PMD \\
\bottomrule
\end{tabular}
\end{minipage}
&
\begin{minipage}[t]{0.3\textwidth}
\centering
\fontsize{7pt}{9pt}\selectfont
\textbf{RGG(200, 0.2)} \\
\begin{tabular}{lccc}
\toprule
Method & $p$-value & $A$ & Winner \\
\midrule
GGC     & $\ll 0.01$ & $0.00$ & GGC \\
Genetic & $\ll 0.01$ & $1.00$ & GP-PMD \\
BQP     & $\ll 0.01$ & $1.00$ & GP-PMD \\
MA-NWS  & $\ll 0.01$ & $1.00$ & GP-PMD \\
GP-PGP  & $\ll 0.01$ & $1.00$ & GP-PMD \\
\bottomrule
\end{tabular}
\end{minipage}
&
\begin{minipage}[t]{0.3\textwidth}
\centering
\fontsize{7pt}{9pt}\selectfont
\textbf{RGG(200, 0.3)} \\
\begin{tabular}{lccc}
\toprule
Method & $p$-value & $A$ & Winner \\
\midrule
GGC     & $\ll 0.01$ & $0.00$ & GGC \\
Genetic & $\ll 0.01$ & $1.00$ & GP-PMD \\
BQP     & $\ll 0.01$ & $1.00$ & GP-PMD \\
MA-NWS  & $\ll 0.01$ & $1.00$ & GP-PMD \\
GP-PGP  & $\ll 0.01$ & $0.99$ & GP-PMD \\
\bottomrule
\end{tabular}
\end{minipage}
\\[1em]

\begin{minipage}[t]{0.3\textwidth}
\centering
\fontsize{7pt}{9pt}\selectfont
\textbf{RGG(500, 0.1)} \\
\begin{tabular}{lccc}
\toprule
Method & $p$-value & $A$ & Winner \\
\midrule
GGC     & $\ll 0.01$ & $0.00$ & GGC \\
Genetic & $\ll 0.01$ & $1.00$ & GP-PMD \\
BQP     & $\ll 0.01$ & $1.00$ & GP-PMD \\
MA-NWS  & $\ll 0.01$ & $1.00$ & GP-PMD \\
GP-PGP  & $\ll 0.01$ & $1.00$ & GP-PMD \\
\bottomrule
\end{tabular}
\end{minipage}
&
\begin{minipage}[t]{0.3\textwidth}
\centering
\fontsize{7pt}{9pt}\selectfont
\textbf{RGG(500, 0.2)} \\
\begin{tabular}{lccc}
\toprule
Method & $p$-value & $A$ & Winner \\
\midrule
GGC     & $\ll 0.01$ & $0.00$ & GGC \\
Genetic & $\ll 0.01$ & $1.00$ & GP-PMD \\
BQP     & $\ll 0.01$ & $1.00$ & GP-PMD \\
MA-NWS  & $\ll 0.01$ & $1.00$ & GP-PMD \\
GP-PGP  & $\ll 0.01$ & $0.94$ & GP-PMD \\
\bottomrule
\end{tabular}
\end{minipage}
&
\begin{minipage}[t]{0.3\textwidth}
\centering
\fontsize{7pt}{9pt}\selectfont
\textbf{RGG(500, 0.3)} \\
\begin{tabular}{lccc}
\toprule
Method & $p$-value & $A$ & Winner \\
\midrule
GGC     & $\ll 0.01$ & $0.00$ & GGC \\
Genetic & $\ll 0.01$ & $0.90$ & GP-PMD \\
BQP     & $\ll 0.01$ & $1.00$ & GP-PMD \\
MA-NWS  & $\ll 0.01$ & $1.00$ & GP-PMD \\
GP-PGP  & $\ll 0.01$ & $0.83$ & GP-PMD \\
\bottomrule
\end{tabular}
\end{minipage}
\end{tabular}
\caption{GP-PMD versus other algorithms on computational time}
\label{tab:testing_time}
\end{figure}

\end{document}